\documentclass[11pt,letterpaper]{article}
\usepackage[margin=1in]{geometry}

\usepackage{cite}
\usepackage{amsmath,amssymb,amsfonts,amsthm}
\usepackage{algorithmic}
\usepackage{graphicx,multicol}
\usepackage{textcomp}
\def\BibTeX{{\rm B\kern-.05em{\sc i\kern-.025em b}\kern-.08em
    T\kern-.1667em\lower.7ex\hbox{E}\kern-.125emX}}
\usepackage{epstopdf}
\usepackage{caption}
\usepackage{url}
\usepackage{cite}
\usepackage{color}
\usepackage{subcaption}
\usepackage{marvosym}

\usepackage{tikz}
\usepackage{multirow}
\usetikzlibrary{shapes,arrows,calc}
\usetikzlibrary{positioning}
\usetikzlibrary{arrows,automata}

\usepackage{amsmath} 
\usepackage{dsfont}
\usepackage{standalone}
\usepackage{booktabs, caption, array}
\usepackage{siunitx}
\usepackage{epstopdf}
\usepackage{epsfig} 
\usepackage{subcaption}
\usepackage{psfrag}
\usepackage{lipsum}
\usepackage{amssymb}  
\usepackage{cite}
\usepackage{wrapfig}
\usepackage{mathrsfs}
\usepackage{url}
\usepackage{color}
\usepackage{colortbl}

\usepackage[authoryear]{natbib}
\usepackage[pdfpagelabels,pdfpagemode=None,breaklinks=true]{hyperref} \hypersetup{
     colorlinks   = true,
     citecolor    = blue
}
\renewcommand{\cite}{\citep}

\newtheorem{theorem}{Theorem}

\newtheorem{lemma}{Lemma}
\newtheorem{remark}{Remark}

\newtheorem{proposition}{Proposition}

\allowdisplaybreaks

\newcommand{\ignore}[1]{}

\newcommand{\real}{\mathbb{R}}

\newcommand{\SUt}{\mathtt{SU}}
\newcommand{\SPt}{\mathtt{SP}}
\newcommand{\IPt}{\mathtt{IP}}
\newcommand{\IUt}{\mathtt{IU}}
\newcommand{\RPt}{\mathtt{RP}}
\newcommand{\RUt}{\mathtt{RU}}
\newcommand{\betau}{\beta_\mathtt{U}}
\newcommand{\betap}{\beta_\mathtt{P}}
\newcommand{\cp}{c_\mathtt{P}}

\newcommand{\hatbetau}{\hat{\beta}_\mathtt{U}}
\newcommand{\hatbetap}{\hat{\beta}_\mathtt{P}}

\newcommand{\subscr}[2]{#1_{\textup{#2}}}

\newcommand{\oprocendsymbol}{\hbox{$\bullet$}}
\newcommand{\oprocend}{\relax\ifmmode\else\unskip\hfill\fi\oprocendsymbol}

\allowdisplaybreaks

\def \mr{\mathrm}

\begin{document}

\title{Coupled Evolutionary Behavioral and Disease Dynamics under Reinfection Risk}

\author{Abhisek Satapathi, Narendra Kumar Dhar, Ashish R. Hota~and~Vaibhav Srivastava\thanks{Abhisek Satapathi and Ashish R. Hota are with the Department of Electrical Engineering, IIT Kharagpur, India. Narendra Kumar Dhar is with the School of Computing and Electrical Engineering, IIT Mandi, India. Vaibhav Srivastava is with the Department of Electrical and Computer Engineering, Michigan State University, East Lansing, MI 48824. (Email: abhisek.ee@iitkgp.ac.in, narendra@iitmandi.ac.in, ahota@ee.iitkgp.ac.in, vaibhav@msu.edu). This work was supported in part by ARO grant W911NF-18-1-0325. A preliminary version of this work \cite{satapathi2022epidemic} appeared at the American Control Conference, 2022.}}%

\date{}
\maketitle

\begin{abstract}
We study the interplay between epidemic dynamics and human decision making for epidemics that involve reinfection risk; in particular, the susceptible-infected-susceptible (SIS) and the susceptible-infected-recovered-infected (SIRI) epidemic models. In the proposed game-theoretic setting, individuals choose whether to adopt protection or not based on the trade-off between the cost of adopting protection and the risk of infection; the latter depends on the current prevalence of the epidemic and the fraction of individuals who adopt protection in the entire population. We define the coupled epidemic-behavioral dynamics by modeling the evolution of individual protection adoption behavior according to the replicator dynamics. For the SIS epidemic, we fully characterize the equilibria and their stability properties. We further analyze the coupled dynamics under timescale separation when individual behavior evolves faster than the epidemic, and characterize the equilibria of the resulting discontinuous hybrid dynamical system for both SIS and SIRI models. Numerical results illustrate how the coupled dynamics exhibits oscillatory behavior and convergence to sliding mode solutions under suitable parameter regimes. 
\end{abstract}

\section{Introduction}

Infectious diseases or epidemics spread through human society via social interactions among infected and healthy individuals. There are two broad classes of epidemic models: (i) epidemics where an individual upon recovery is no longer at risk of future infection (e.g., the susceptible-infected-recovered (SIR) epidemic model) and (ii) epidemics where recovery from infection does not lead to immunity from future infections \cite{mei2017dynamics,nowzari2016analysis}. Motivated by a large number of infectious diseases where recovered individuals are still at risk of future infection, including COVID-19 \cite{crellen2021dynamics}, we focus on the second class of epidemics in this work; in particular, the susceptible-infected-susceptible (SIS) and the susceptible-infected-recovered-infected (SIRI) \cite{pagliara2018bistability} epidemic models. In the SIRI model, the infection rate of individuals who have been infected in the past is different from those who have never been infected which captures diseases that impart compromised or strengthened immunity after initial infection. 

In the absence of suitable medicines or vaccines, particularly at the onset of an epidemic, individuals adopt protective measures (such as wearing masks or maintaining social distancing) in order to avoid becoming infected, often in a strategic and decentralized manner. Thus, adoption of protective behavior both affects and is affected by the level of infection prevalence. Consequently, selfish adoption of protective behavior has been studied in the framework of game theory \cite{chang2020game,huang2021game}. For the SIS epidemic, both static or single-shot games to model vaccination decisions \cite{hota2019game,trajanovski2015decentralized} and dynamic games that model evolution of protective decisions \cite{theodorakopoulos2013selfish,eksin2016disease,huang2019differential,hota2020impacts} have been analyzed in recent past. Specifically, in \cite{theodorakopoulos2013selfish}, individuals in a well-mixed population decide whether or not to adopt protection (which eliminates the risk of infection) at a certain cost, and epidemic evolution under optimal protection strategies is studied. This setting is generalized in \cite{eksin2016disease} and \cite{hota2020impacts} to include network interactions. The COVID-19 pandemic led to renewed interest in this topic, with several recent papers exploring game-theoretic protection and vaccination strategies for the class of SIR epidemic and its variants \cite{amini2022epidemic,kordonis2022dynamic,altman2022mask}. While these works characterize Nash equilibrium strategies under various assumptions on the payoffs and population heterogeneity, they do not explore (i) (evolutionary) learning dynamics for the decision-makers and (ii) equilibria (and their stability) of the epidemic dynamics under game-theoretic strategies.

In particular, the problem of a large group of strategic individuals finding or learning equilibrium strategies is quite challenging \cite{daskalakis2009complexity,sandholm2010population}. In settings with a large population of agents, {\it evolutionary dynamics} have been proposed and their convergence behavior to the set of equilibria have been analyzed \cite{hofbauer2003evolutionary,sandholm2010population}. This class of learning dynamics involves repeated play of a static game where the payoff functions remain unchanged. In contrast with the above settings, epidemic games exhibit a dynamically changing proportion of healthy and infected individuals. A recent work \cite{elokda2021dynamic} investigated this dynamic evolution of infection dynamics and strategic decisions for the class of susceptible-asymptomatic-infected-recovered (SAIR) epidemic in the framework of dynamic population games. Coupled evolution of disease and behavior have also been investigated in \cite{martins2022epidemic,khazaei2021disease} for SIRS and SEIR epidemic models, respectively. While the setting in \cite{martins2022epidemic} allows for reinfection, the payoff function of an individual does not directly depend on the infection prevalence, rather the payoffs are governed by a dynamic reward term likely designed by a central authority. The authors in \cite{liu2022herd} consider a coupled SIS epidemic and evolutionary learning model based on mean dynamic and analyze equilibria when disease dynamics is faster than learning dynamics. In a contemporary work \cite{frieswijk2022mean}, the authors consider a coupled SIS epidemic-behavioral setting in which the payoff contains a social influence factor and imitation dynamics is used to model the evolution of behavior. 

In this paper, we build upon the above line of work, and consider the SIS and SIRI epidemic\footnote{To the best of our knowledge, game-theoretic analysis of protection strategies for the SIRI epidemic model has yet not been reported.} settings where a large population of individuals choose whether to adopt protection or to remain unprotected as the epidemic evolves. Adopting protection reduces the infection probability for healthy individuals and transmission probability for infected individuals.\footnote{Our formulation, with partially effective protection, generalizes the setting in \cite{theodorakopoulos2013selfish} where adopting protection completely eliminated the infection risk.} The cost of adopting protection is weighed with the instantaneous risk of becoming infected; the latter depends on the current epidemic prevalence and the proportions of individuals in different epidemic states adopting protection. We focus on the {\it replicator dynamics} \cite{hofbauer2003evolutionary,sandholm2010population} to model the evolution of protection decisions and study their interaction with SIS and SIRI dynamics. Our choice is motivated by the fact that the replicator dynamics is one of the most well studied evolutionary learning dynamics and has seen widespread applications in biological, environmental and socio-economic settings \cite{cressman2014replicator,sun2021comparison,weitz2016oscillating,wang2017testability,kurita2022covid}. The major contributions of this work are stated below. 

First, for the coupled SIS epidemic and replicator dynamics, we completely characterize the equilibria (existence and local stability) and show how the stability of different equilibrium points get exchanged as certain parameters change in Section~\ref{section:model}. We further explore the behavior of the coupled dynamics under timescale separation leading to a slow-fast dynamical system \cite{berglund2006noise}. Second, for the coupled SIRI epidemic and replicator dynamics, we analyze possible equilibria under timescale separation in Section~\ref{section:SIRI}. For both (SIS and SIRI) coupled dynamics, we specifically focus on the case where the replicator dynamics are faster.\footnote{This is complementary to the setting in \cite{liu2022herd} which studied the case where disease dynamics is faster than learning dynamics.} We characterize the asymptotic convergence of the infected proportion under game-theoretic strategies to an equilibrium by analyzing the resulting slow dynamics which takes the form of a variable structure system in the limiting regime. We also show that the infected proportion may converge to an endemic sliding mode of the hybrid dynamics for a certain range of parameters. Finally, we provide insights into the transient behavior of these coupled dynamics using numerical simulations in Section~\ref{section:numerical}. Particularly, we illustrate the impact of the timescale separation parameter and oscillatory convergence to the equilibrium point. 
\section{Coupled SIS Epidemic and Evolutionary Behavioral Model}
\label{section:model}

In this section, we formally introduce the coupled evolution of the SIS epidemic and protection adoption  behavior in a homogeneous large-population setting. Let the proportion of susceptible and infected individuals be $s(t)$ and $y(t)$, respectively. Both $s(t), y(t) \in [0,1]$ with $s(t) + y(t) = 1$ for all $t \geq 0$. We adopt a {\it population game} framework \cite{sandholm2010population} where individuals choose whether to adopt protection against the epidemic or not; these actions are denoted by $\mathtt{P}$ and $\mathtt{U}$. Consequently, the {\it population state} at time $t$ is defined as $x(t) := [x_{\SUt}(t) \quad x_{\SPt}(t) \quad x_{\IUt}(t) \quad x_{\IPt}(t)]^\intercal \in \Delta_4$, where $\Delta_n$ is the probability simplex in $\real^n$, $x_{\SUt}$ denotes the proportion of individuals who are susceptible and choose to remain unprotected, $x_{\IPt}$ denotes the proportion of infected individuals who adopt protection, and so on. At time $t$, we have $x_{\SUt}(t) + x_{\SPt}(t) = s(t)$, $x_{\IUt}(t) + x_{\IPt}(t) = y(t)$, and $\mathds{1}^\intercal x(t) = 1$ where $\mathds{1}$ is a vector of appropriate dimension with all entries $1$. Thus, $x(t)$ encodes the joint strategies of the entire population of agents.

Individuals choose their action to maximize their payoffs which depend on the population state $x(t)$ (i.e., the joint strategy profile, which also includes information regarding infection prevalence). For an infected individual, there is no further risk of infection, and as a result, we define its payoff to be constant parameters given by $-c_{\IUt}$ if it remains unprotected and $- c_{\IPt}$ if it adopts protection. For instance, $c_{\IUt} > 0$ captures the cost of breaking isolation/quarantine protocols for an infected individual. Thus, we assume $c_{\IUt} > c_{\IPt} \geq 0$.

A susceptible individual trades off the cost of adopting protection, denoted by $\cp > 0$, and expected cost of becoming infected. The latter is the product of the loss upon infection $L > 0$\footnote{The value of $L$ may arise from various considerations, including perception or opinion regarding disease severity prevalent in the society, expected long-run economic loss for the duration of the disease, among others. Further explorations along these lines lies beyond the scope of this paper. In this work, we focus on equilibrium behavior for a given known value of $L$.} and the instantaneous probability of becoming infected which depends on its action and the strategies of other agents (captured in the population state). Specifically, let $\betau$ and $\betap$ denote the probabilities of an infected individual causing a new infection if it is unprotected and protected, respectively. We impose the natural assumption $\betau > \betap \geq 0$ throughout the paper. Consequently, the instantaneous probability of infection for an unprotected susceptible individual at population state $x$ is given by $\betau x_{\IUt} + \betap x_{\IPt}$. Similarly, let a protected susceptible individual be $\alpha \in (0,1)$ times (less) likely to become infected compared to an unprotected susceptible individual. Thus, the instantaneous probability of infection for a protected susceptible individual at population state $x$ is given by $\alpha(\betau x_{\IUt} + \betap x_{\IPt})$. The payoff vector of individuals at population state $x$ is now defined as
\begin{equation}\label{eq:payoff_sis}
F(x) \!:= 
\begin{bmatrix}
F_{\SUt}(x) \\
F_{\SPt}(x) \\
F_{\IUt}(x) \\
F_{\IPt}(x)
\end{bmatrix}
\!=\!
\begin{bmatrix}
-\! L (\betau x_{\IUt} + \betap x_{\IPt}) \\
-\cp -\! L \alpha (\betau x_{\IUt} + \betap x_{\IPt}) \\
- c_{\IUt} \\
- c_{\IPt}
\end{bmatrix}
,
\end{equation}
where $F_{\SUt}(x)$ denotes the payoff for an individual who is susceptible and unprotected at population state $x$ (and thus depends on the joint strategy profile), and so on. Susceptible individuals who adopt protection pay a cost $\cp$ but experience a reduced infection risk scaled by factor $\alpha$ as discussed above. 

While most of the past works have proceeded to directly analyze the equilibrium strategies after introducing the payoff function, we here consider the evolutionary learning dynamics under which individual protection decisions evolve. We first introduce some notation. Let $z_S(t) \in [0,1]$ denote the fraction of susceptible individuals who remain unprotected, i.e., $x_{\SUt}(t) = z_S(t) s(t)$ and $x_{\SPt}(t) = (1-z_S(t)) s(t)$. Similarly, let $z_I(t) \in [0,1]$ denote the fraction of infected individuals who remain unprotected. Due to the presence of both unprotected and protected individuals with different infection probabilities, the infected proportion evolves as
\begin{align}
\dot{y}(t) & = (x_{\SUt}(t) + \alpha x_{\SPt}(t)) (\betau x_{\IUt}(t) + \betap x_{\IPt}(t)) - \gamma y(t) \nonumber
\\ & = \big[(1-y(t)) (z_S(t) + \alpha (1-z_S(t))) (\betau z_I(t) + \betap (1-z_I(t))) - \gamma \big] y(t) \nonumber
\\ & =: f_y(y(t),z_S(t),z_I(t)), \label{eq:sis_scalar_com}
\end{align}
where $\gamma$ is the rate of recovery for infected individuals. The above dynamics is analogous to the conventional scalar SIS epidemic dynamics with effective infection rate $\subscr{\beta}{eff}(z_S,z_I) = (z_S + \alpha (1-z_S)) (\betau z_I + \betap (1-z_I))$ which now depends on the efficacy of protection and the fractions that adopt protection. If $\alpha = 1$ and $\betap = \betau = \beta$, i.e., protection is not effective, the effective infection rate is $\beta$ which is the setting in classical SIS epidemic without protection. Note further that in contrast with the classical SIS epidemic setting, $\subscr{\beta}{eff}(z_S,z_I)$ is time-varying as the fractions of susceptible and infected individuals adopting protection evolves with time in accordance with evolutionary learning dynamics as discussed below.

We focus on the class of replicator dynamics~\cite{sandholm2010population,cressman2014replicator} in this work and assume that susceptible individuals only replicate the strategies of other susceptible individuals (likewise for infected individuals). For susceptible individuals, we obtain
\begin{align}
\dot{z}_{S}(t) & =  {z}_{S}(t)(1-{z}_{S}(t)) \big[ F_{\SUt}(x) - F_{\SPt}(x) \big] \nonumber
\\ & = {z}_{S}(t)(1-{z}_{S}(t))\big[ \cp - L(1-\alpha)(\betau z_{I}(t) + \betap (1-z_I(t))) y(t) \big] \nonumber
\\ & =: f_S(y(t),z_S(t),z_I(t)). \label{eq:main_zs}
\end{align}
Similarly, for infected individuals, we have
\begin{align}
\dot{z}_{I}(t) 
 & = {z}_{I}(t) (1-{z}_{I}(t)) (c_{\IPt}-c_{\IUt}) =: f_I(y(t),z_S(t),z_I(t)). \label{eq:main_zi}
\end{align}

Thus, equations \eqref{eq:sis_scalar_com}, \eqref{eq:main_zs} and \eqref{eq:main_zi} characterize the coupled evolution of the epidemic and population states which remain confined to the set $[0,1]^3$ as shown below.

\begin{lemma}[Invariant set of Coupled SIS Epidemic]
\label{lem:invariant-set}
For the coupled SIS epidemic and evolutionary behavioral dynamics defined by \eqref{eq:sis_scalar_com}, \eqref{eq:main_zs} and \eqref{eq:main_zi}, the set $\{(y, z_S, z_I) \rvert (y, z_S, z_I) \in [0,1]^3\}$ is invariant. 
\end{lemma}
\begin{proof}
It is easy to see that when $y(t) =0$, $\dot y = 0$ and when $y(t) =1$, $\dot y < 0$. Likewise, for $z_I, z_S \in \{0,1\}$, $\dot z_I = \dot z_S =0$. The result follows from Nagumo's theorem~\cite[Theorem 4.7]{blanchini2008set}.  
\end{proof}

\subsection{Equilibrium Characterization and Stability Analysis}

We now examine the equilibrium points of the above coupled epidemic-replicator dynamics and their stability properties. First, we consider the evolution of $z_I$ in \eqref{eq:main_zi} which does not depend on $y$ and $z_S$. There are two stationary points $z_I = 0$ and $z_I = 1$, and it is easy to see that, for $c_{\IUt} > c_{\IPt}$, $z_I = 1$ is unstable and $z_I = 0$ is exponentially stable with basin of attraction $[0,1)$. It is also quite intuitive that if $c_{\IPt} < c_{\IUt}$, infected individuals prefer to use protection and the strategy for infected individuals should converge to it. 

Thus, in the remainder of this section, we only focus on equilibria with $z_I = 0$. We begin with introducing a few variables that will be used to define the equilibrium points:
\begin{multline*}
y^*_{\mathtt{U}} : = 1 - \frac{\gamma}{\betap}, \quad
y^*_{\mathtt{P}} := 1 - \frac{\gamma}{\alpha \betap}, \quad 
y^*_{\mr{int}}:= \frac{\cp}{L(1-\alpha)\betap} \text{and }  z^*_{S,\mr{int}} := \frac{1}{1-\alpha} \left[ \frac{\gamma}{\betap(1-y^*_{\mr{int}})} - \alpha \right].
\end{multline*}
	
	We now define all possible equilibria $(y^*,z^*_S,z^*_I)$ of the coupled SIS epidemic and evolutionary behavior dynamics  (\ref{eq:sis_scalar_com}--\ref{eq:main_zi}) corresponding to $z_I=0$:
	\begin{multline*}
		\mathbf{E0}= (0,0,0), \quad 	\mathbf{E1}= (0,1,0),  \quad  \mathbf{E2}= (y^*_{\mathtt{U}},1,0), 
		\mathbf{E3}= (y^*_{\mr{int}},z^*_{S,\mr{int}},0),  \text{ and}  \quad  \mathbf{E4}= (y^*_{\mathtt{P}},0,0). 
	\end{multline*}
	
At $\mathbf{E0}$ everyone adopts protection and there is no infection. At $\mathbf{E1}$, there is no infection, and susceptible individuals do not adopt protection. $ \mathbf{E2}$ is an endemic equilibrium, i.e., a fraction of the population is infected, and susceptible individuals continue to remain unprotected. $\mathbf{E3}$ is an endemic equilibrium where a fraction of susceptible individuals adopt protection. Finally, $\mathbf{E4}$ is an endemic equilibrium  at which all  susceptible individuals adopt protection. The existence and local stability of these equilibria, as summarized in Table \ref{tab:eq_summary}, are established in the following proposition whose proof is presented in Appendix \ref{app:sis:1}. 

\begin{table*}[t]
\centering
	{\renewcommand{\arraystretch}{1.5}
\caption{Existence and stability of equilibria of the coupled dynamics \eqref{eq:sis_scalar_com}, \eqref{eq:main_zs} and \eqref{eq:main_zi}.}
\label{tab:eq_summary}
		\begin{tabular}{|c |c |c |c |c |c |}
			\hline
			\multirow{2}{*}{Epidemic} & \multirow{2}{*}{Endemic} & \multicolumn{4}{c|}{Equilibria} \\ \cline{3-6} 
			
			Parameters & Infection Level & $\mathbf{E1}: (0,1,0)$ & $\mathbf{E2}: (y^*_u,1,0)$ & $\mathbf{E3}: (y^*_{\mr{int}},z^*_{S,int},0)$ & $\mathbf{E4}: (y^*_p,0,0)$ \\ \hline \hline
			
			$\gamma > \beta_p$ & $-$ & \cellcolor{green!45} $\checkmark$, stable & $-$ & $-$ & $-$ \\ \hline
			
			\multirow{2}{*}{$\alpha \beta_p < \gamma < \beta_p$} & $y^*_u < y^*_{\mr{int}}$ & $\checkmark$, unstable & \cellcolor{green!45} $\checkmark$, stable  & $-$  & $-$ \\ \cline{2-6} 
			
			& $y^*_{\mr{int}} < y^*_u$ & $\checkmark$, unstable & $\checkmark$, unstable & \cellcolor{green!45} $\checkmark$, stable & $-$ \\ \hline
			
			\multirow{3}{*}{$\gamma < \alpha \beta_p$} & $y^*_{u} < y^*_{\mr{int}}$ & $\checkmark$, unstable & \cellcolor{green!45} $\checkmark$, stable &  $-$ & $\checkmark$, unstable \\ \cline{2-6} 
			
			& $y^*_p < y^*_{\mr{int}} < y^*_u$ & $\checkmark$, unstable & $\checkmark$, unstable  & \cellcolor{green!45} $\checkmark$, stable & $\checkmark$, unstable \\ \cline{2-6}
			
			& $y^*_{\mr{int}} < y^*_p$ & $\checkmark$, unstable & $\checkmark$, unstable & $-$ & \cellcolor{green!45} $\checkmark$, stable \\ \hline
			
	\end{tabular}}
\end{table*}

\medskip

\begin{proposition}[Equilibria and Stability]\label{prop:equilibria-and-stability}
For the equilibrium points  of the coupled SIS epidemic and evolutionary behavioral dynamics  (\ref{eq:sis_scalar_com}--\ref{eq:main_zi}) corresponding to $z_I=0$, the following statements hold: 
\begin{enumerate}
\item $\mathbf{E0}$ exists for all parameter regimes,  and is unstable; 
\item $\mathbf{E1}$ exists for all parameter regimes, is locally stable if $\betap < \gamma$, and is unstable, otherwise; 
\item $ \mathbf{E2}$ exists only when $\beta_p > \gamma$, is locally stable when $y^*_{\mathtt{U}} < y^*_{\mr{int}}$, and is unstable otherwise; 
\item $\mathbf{E3}$ exists only when $\gamma < \betap$ and $ y^*_{\mathtt{P}}< y^*_{\mr{int}} <  y^*_{\mathtt{U}}$, and is locally stable; and
\item  $\mathbf{E4}$ exists only when $\gamma < \alpha \betap$, is locally stable when $y^*_{\mathtt{P}} > y^*_{\mr{int}}$, and is unstable otherwise. 
\end{enumerate}
\end{proposition}

\begin{remark}
While the above result shows local stability of the equilibrium points, we can show global attractivity if we restrict the coupled dynamics to the $(y,z_S)$ plane, i.e., by setting $z_I=0$ for the planar  dynamics \eqref{eq:sis_scalar_com} and \eqref{eq:main_zs}. Note that all equilibrium points except $\mathbf{E3}$ lie on the boundary of the invariant set $[0,1]^2$ for the planer dynamics. Thus, when $\mathbf{E3}$ does not exist (i.e., in all regimes other than 4) in Proposition~\ref{prop:equilibria-and-stability}, it follows from index theory \cite[Section 6.8]{strogatz2018nonlinear} that no limit cycle exists. Since at any given set of parameter values, there is exactly one equilibrium point which is locally stable, it is globally attractive as well.
\end{remark}

\begin{remark}
As the recovery rate $\gamma$ increases, the infected proportion at the stable equilibrium decreases. In practice, $\gamma$ is improved via direct intervention of authorities by augmenting healthcare facilities. When further resource augmentation is not possible, a more effective protection scheme with a smaller value of $\alpha$ would result in a smaller value of $y^*_{\mathtt{P}}$. In addition, if $\alpha$ is reduced further such that $\gamma > \alpha \betap$, the dynamics exhibits a new stable endemic infection level $y^*_{\mr{int}}$ which also decreases in $\alpha$. Thus, our result shows that indirect intervention by facilitating availability of more effective protection schemes would significantly contribute towards a smaller endemic infection level under game-theoretic strategies.
\end{remark}

\begin{remark}
While the above result holds under the assumption $c_{\IUt} > c_{\IPt}$, the results for the case $c_{\IUt} < c_{\IPt}$ are analogous with $z_I = 1$ being the stable equilibrium for the infected population. In fact, when the proportion of infected agents adopting protection is a constant $z^\star_{I}$, then the results presented above would continue to hold by redefining $\beta_{\mathtt{P}} := \beta_{\mathtt{U}} z^\star_{I} + \beta_{\mathtt{P}} (1-z^\star_{I})$. Similarly, when $\alpha=0$, $y^*_{\mathtt{P}} = -\infty$ and $\mathbf{E4}$ will cease to exist as an equilibrium point. The equilibrium behavior of the coupled dynamics will continue to be governed by the first four cases of the above proposition.
\end{remark}

\subsection{Bifurcation Analysis}\label{sec:bif-analysis}

The above proposition shows that the equilibrium points exchange stability properties as certain parameters, e.g., $\gamma, \betap, \alpha$, vary. We now numerically explore the bifurcations associated with the transition of stability among the equilibria. For the numerical illustration we choose the parameter values in equations \eqref{eq:sis_scalar_com}, \eqref{eq:main_zs} and \eqref{eq:main_zi} as  summarized below. 

\begin{center}
	\begin{tabular}{|c|c|c|c|c|c|c|}
	\hline
	$\cp$ & 	$\alpha$ & $\betau$ &	$c_{\IUt}$ & $L$ &  $\betap$ & $c_{\IPt}$  \\ \hline 
	$1$ & $0.5$ & $0.3$ & $2$ & $80$ & $0.15$ & $1$ \\ \hline 
\end{tabular} 
\end{center}

We adopt the recovery rate $\gamma$ as a bifurcation parameter and use the numerical continuation package MATCONT~\cite{dhooge2003matcont} to compute the bifurcation diagram shown in Fig.~\ref{fig:bifurcation-diagram}.

\begin{figure}[ht!]
	\centering 
\includegraphics[width=0.8\linewidth]{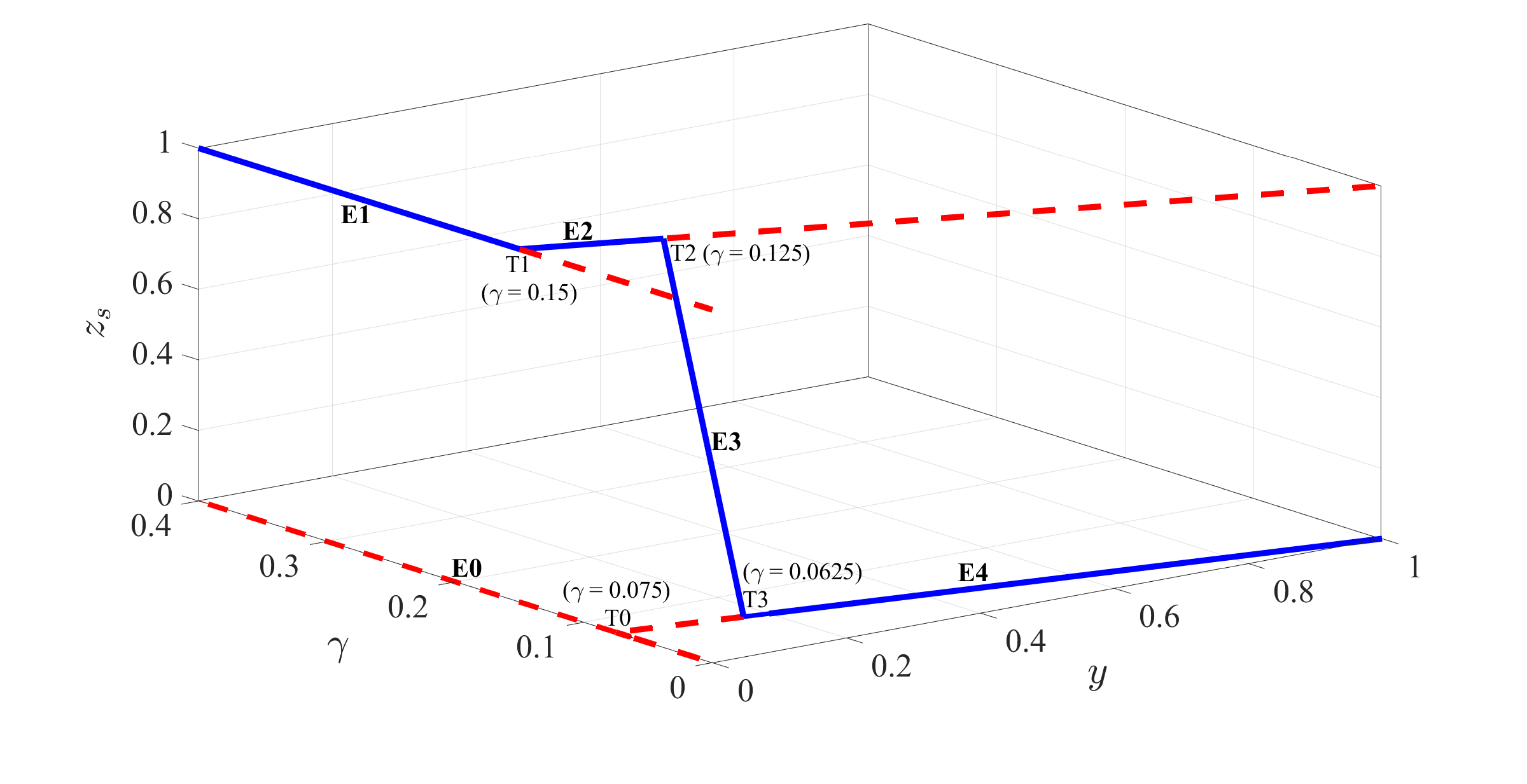} \\
\caption{Bifurcation diagram of the equilibria of the coupled SIS epidemic and evolutionary behavioral dynamics. (solid blue: stable branch and dashed red: unstable branch of equilibria)} \label{fig:bifurcation-diagram}
\end{figure}

For $\gamma \to 0^+$, $\mathbf{E4}$ is the stable equilibrium, while $\mathbf{E0}$, $\mathbf{E1}$, and $\mathbf{E2}$ are unstable. As the value of $\gamma $ is increased at point T3 in  Fig.~\ref{fig:bifurcation-diagram}, $\mathbf{E4}$  exchanges stability to $\mathbf{E3}$ in a transcritical bifurcation. Note that unstable branch of $\mathbf{E3}$ is not visible since it is associated with negative values of $z_S$. As $\gamma$ is increased, the fraction of susceptible population that adopts protection decreases and at T2, $\mathbf{E3}$ exchanges stability to $\mathbf{E2}$ in another transcritical bifurcation. Again, the unstable branch of $\mathbf{E3}$ at T2 corresponds to $z_s>1$ and is not visible in the diagram. Upon further increasing $\gamma$, the fraction of infected population continues to decrease and at  T1, $\mathbf{E2}$ exchanges stability with the disease-free equilibria $\mathbf{E1}$. The unstable branch of $\mathbf{E2}$ at T1 corresponds to negative values of $y$. 

Another transcritical bifurcation takes place at T0 $(\gamma= \alpha \betap)$, where $\mathbf{E0}$ and $\mathbf{E4}$ cross.   For $\gamma < \alpha \betap$ (resp., $\gamma > \alpha \betap$), $\mathbf{E0}$ has two (resp., one) eigenvalues in the right-half plane. As $\mathbf{E4}$ approaches T0 from $y >0$, it has one eigenvalue in the right-half plane, while for $y<0$ near T0, $\mathbf{E4}$ is stable. Thus, the transcritical bifurcation at T0 corresponds to exchange of stable and unstable eigenvalues of two unstable equilibria.

It can be verified that each of the above bifurcations are indeed transcritical. For example, at $\gamma = \betap$, $\mathbf{E1}$ and $\mathbf{E2}$ exchange stability. It follows from the proof of Proposition~\ref{prop:equilibria-and-stability} that at $\gamma = \betap$,  $J_{\mathbf{E1}} = J_{\mathbf{E2}}$ has only one eigenvalue at zero and the associated left and right eigenvalues are $\boldsymbol e_1^\top$ and $\boldsymbol e_1$, where $\boldsymbol e_1 = [1\; 0\; 0]^\top$. Let $\boldsymbol \xi = [y \; z_S \; z_I]$ and $\boldsymbol f (\boldsymbol \xi ) = [f_y (\boldsymbol \xi )\; f_S (\boldsymbol \xi )\;  f_I(\boldsymbol \xi )]^\top$. Then, it can be verified that $\boldsymbol e_1^\top (\partial^2 \boldsymbol f  / \partial \gamma \partial \boldsymbol \xi ) (\boldsymbol e_1) =-1 \ne 0$ at $(\gamma, y, z_S, z_I) = (\betap, 0, 1, 0)$.  Additionally,  $\boldsymbol e_1^\top (\partial^2 \boldsymbol f  / \partial^2 \boldsymbol \xi (\boldsymbol e_1, \boldsymbol e_1) )  =-2 \alpha \betap \ne 0$ at $(\gamma, y, z_S, z_I) = (\betap, 0, 1, 0)$. Thus, the bifurcation at $\gamma = \betap$ is transcritical~\cite[Section 3.4]{JG-PH:90}. The other bifurcations can be analyzed similarly.


\subsection{Coupled Epidemic-Behavioral Dynamics under Timescale Separation}\label{sec:timescale-separation}

We have thus far assumed that the the epidemic and the replicator dynamics evolve at the same time-scale. However, it is not strictly necessary for the coupled dynamics. In order to obtain further insights into their behavior, we now study the coupled epidemic-behavioral dynamics \eqref{eq:sis_scalar_com}, \eqref{eq:main_zs} and \eqref{eq:main_zi} under timescale separation. In particular, we focus on the case in which the replicator dynamics evolves faster than the epidemic dynamics; indeed in the modern era, there is an increased awareness about infectious diseases due to publicly available testing data, awareness campaigns by public health authorities, spread of information via social media, which shapes human response at a much faster time-scale. To this end, we model the coupled dynamics as a slow-fast system:
\begin{align} \label{eq:coupled-dynamics-timescale_slow_epi}
\begin{split}
\dot{y}(t) & =  f_y(y(t),z_S(t),z_I(t)) \\
\epsilon \dot{z}_{S}(t) & =  f_S(y(t),z_S(t),z_I(t)) \\
\epsilon \dot{z}_{I}(t) & = f_I(y(t),z_S(t),z_I(t)), 
\end{split}
\end{align}
where $\epsilon \in (0,1]$ is a timescale separation variable \cite{berglund2006noise}.

At a given epidemic prevalence $y$, we characterize the (stable) equilibria of the fast system involving the replicator dynamics with states $(z_S,z_I)$. For reasons discussed earlier, we focus on equilibria with $z_I = 0$. It is now easy to see that if $y \neq y^*_{\mr{int}}$, there are two equilibrium points: $(0,0)$ and $(1,0)$. If $y = y^*_{\mr{int}}$, then $(z_S,0)$ is an equilibrium point of the fast system for any $z_S \in [0,1]$. Following analogous arguments as in the proof of Proposition \ref{prop:equilibria-and-stability}, it follows that  $(0,0)$ is locally stable for the fast system when $y > y^*_{\mr{int}}$ and $(1,0)$ is locally stable for the fast system when $y < y^*_{\mr{int}}$. Consequently, we obtain the following reduced dynamics for the slow system which approximates the coupled dynamics \eqref{eq:coupled-dynamics-timescale_slow_epi} in the limit $\epsilon \to 0$ as 
\begin{align}
\dot{y}(t) & = 
\begin{cases}
\big[(1-y(t)) \betap - \gamma \big] y(t), \quad & \text{if } y(t) < y^*_{\mr{int}}, 
\\ \big[(1-y(t)) \alpha \betap - \gamma \big] y(t), \quad & \text{if } y(t) > y^*_{\mr{int}},
\end{cases}
\label{eq:epi_slow}
\\ \text{and if  } y(t) = y^*_{\mr{int}}, \text{we have} \qquad  \dot{y}(t) & \!\in\! \{\big[(1\!-\!y^*_{\mr{int}})\betap (z_S \!+\! \alpha(1\!-\! z_S)) \!-\! \gamma \big] y^*_{\mr{int}} \; \rvert \; z_S \! \in [0,1]\}.  \nonumber
\end{align}


In particular, since the reduced dynamics is an instance of a discontinuous dynamical system, we define the dynamics at the point of discontinuity $y = y^*_{\mr{int}}$ as a differential inclusion which is also the convex combination of the dynamics on both sides of $y = y^*_{\mr{int}}$. It is easy to see that the right hand side of the dynamics is measurable and is locally essentially bounded, and therefore \eqref{eq:epi_slow} admits a Filippov solution \cite[Proposition 3]{cortes2008discontinuous}. We now establish convergence of $y(t)$ under \eqref{eq:epi_slow}. 

\begin{proposition}[Trajectories under fast behavioral response]\label{prop:epidemic_slow}
	For the epidemic dynamics \eqref{eq:epi_slow} with $y(0) \neq 0$,  the following statements hold:
	\begin{enumerate}
		\item if $ y^*_{\mathtt{U}} \leq 0$, then $y(t)$ monotonically decreases and converges to the origin; 
		\item if  $0 < y^*_{\mathtt{U}} < y^*_{\mr{int}}$, then $y(t)$ monotonically converges to $y_u^*$; 
		\item if  $ y^*_{\mathtt{P}} < y^*_{\mr{int}} < y^*_{\mathtt{U}}$, then $y(t)$ converges to $y^*_{\mr{int}}$ which acts as a sliding mode of the dynamics;
		\item if  $y^*_{\mr{int}} < y^*_{\mathtt{P}}$, then $y(t)$ monotonically converges to $y^*_{\mathtt{P}}$.
	\end{enumerate}
\end{proposition}

The proof of the above proposition is presented in Appendix \ref{app:sis:2}. 

\medskip

\begin{remark}
The dynamics in \eqref{eq:epi_slow} potentially represents a class of non-pharmaceutical interventions where authorities impose social distancing measures that reduces the infection rate by a factor $\alpha$ when the infection prevalence exceeds a threshold. Thus, the result in Proposition \ref{prop:epidemic_slow} is potentially of independent interest. Further, the above proposition generalizes analogous results obtained in prior works \cite{theodorakopoulos2013selfish,hota2020impacts} which assumed $\alpha = 0$, i.e., adopting protection completely eliminates risk of infection. 
\end{remark}
\section{Coupled SIRI Epidemic and Evolutionary Behavioral Model}
\label{section:SIRI}

In the SIS epidemic model, a recovered individual encounters the same infection rate as an individual who has never been infected. However, in many infectious diseases, initial infection could lead to compromised immunity \cite{timemag} or it might even lead to reduced risk of future infection \cite{gomes2004infection,crellen2021dynamics}. The susceptible-infected-recovered-infected (SIRI) epidemic model~\cite{pagliara2018bistability}, captures the above characteristics. In this section, we investigate the implications of game-theoretic protection decisions on the evolution of the SIRI epidemic dynamics. 

In the SIRI epidemic model, an individual belongs to one of three possible compartments or states: susceptible, infected and recovered; while the proportion of individuals in each of the above states at time $t$ is denoted by $s(t), y(t)$ and $r(t)$, respectively. The evolution of these proportions is given by
\begin{subequations}
\label{eq:siri_vanilla}
\begin{align}
\dot{s}(t) & = -\beta s(t)y(t) 
\\ \dot{y}(t) & = \beta s(t)y(t) + \hat{\beta} r(t)y(t) - \gamma y(t)
\\ \dot{r}(t) & = \gamma y(t) - \hat{\beta} r(t)y(t), 
\end{align}
\end{subequations}
where $\beta > 0$ is the rate at which susceptible individuals become infected if they encounter an infected individual, $\hat{\beta} > 0$ is the rate at which recovered individuals become infected and $\gamma$ is the rate of recovery for infected individuals. The above transitions are depicted in Fig. \ref{fig:siri_tran}. 

Thus, the above model captures settings where an individual develops immunity upon becoming infected ($\beta > \hat{\beta}$) and when the immunity of an individual is compromised upon infection ($\beta < \hat{\beta}$). When $\beta = \hat{\beta}$ and susceptible and recovered states are combined, we recover the SIS epidemic model. The classical SIR model is also obtained as a special case when $\hat{\beta} = 0$. 


\begin{figure}[tb]
\centering
\begin{tikzpicture}[font=\sffamily]

\tikzset{node style/.style={state, minimum width=1cm, line width=0.3mm, fill=yellow!40!white}}
\tikzset{node style1/.style={state, minimum width=1cm, line width=0.3mm, fill=red!50!white}}
\tikzset{node style2/.style={state, minimum width=1cm, line width=0.3mm, fill=green!50!white}}
\node[node style] at (0, 0) (St)     {$\mathtt{S}$};
\node[node style1] at (3, 0) (At)     {$\mathtt{I}$};
\node[node style2] at (6, 0) (Rt)  {$\mathtt{R}$};
\draw[every loop, auto=right,line width=0.4mm,
              >=latex,
              draw=black,
              fill=black]
(St) edge[bend right=0, auto=left] node[above] {$\beta$} (At)
(At) edge[bend right=0] node[above] {$\gamma$} (Rt)
(Rt) edge[bend right=45] node[above] {$\hat{\beta}$} (At);
\end{tikzpicture}
\caption{\footnotesize Evolution of states in the SIRI epidemic model. Self-loops are omitted for better clarity.}
\label{fig:siri_tran}
\end{figure}
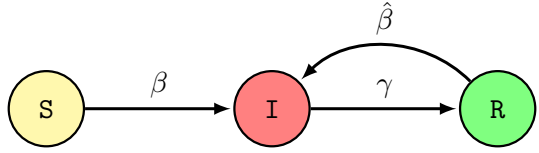

\subsection{Equilibria and Stability without Protective Behavior}

The dynamics in \eqref{eq:siri_vanilla} admits a continuum of equilibria which are infection free (IFE) and may also have an isolated endemic equilibrium (EE) where infection level is nonzero. Mathematically, at the IFE, we have $s = s^*, y = 0, r = 1-s^*$ where $s^* \in [0,1]$. At the EE, we have $s = 0, y = 1-\frac{\gamma}{\hat{\beta}}, r = \frac{\gamma}{\hat{\beta}}$. Before we analyze the implications of strategic adoption of protective behavior, we summarize the main result from \cite{pagliara2018bistability} on the existence and stability of the above equilibria.  

\smallskip

\begin{theorem}[Theorem 2, Lemma 1, Lemma 2 \cite{pagliara2018bistability}]\label{theorem:siri_vanilla}
Let $R_0 := \frac{\beta}{\gamma}$, $R_1 := \frac{\hat{\beta}}{\gamma}$, and $M:=  \frac{1-R_1}{R_0 - R_1}$. Then, the SIRI dynamics exhibits the following behavior. 
\begin{enumerate}
    \item Infection-free: If $R_0<1$ and $R_1<1$, EE is not an equilibrium, while all points in the IFE are locally stable. Further, $y(t)$ decays monotonically to $0$.
    \item Endemic: If $R_0>1$ and $R_1>1$, all points in the IFE are unstable while the EE exists and is locally stable. Further, all solutions reach EE as $t \to \infty$. 
    \item Epidemic: If $R_0>1$ and $R_1\leq1$, then IFE is locally stable if $s^* < M$ and unstable if $s^* > M$. Further, as $t \to \infty$, all solutions reach IFE with $s^* < M$.
    \item Bistable: If $R_0\leq1$ and $R_1>1$, EE exists and is locally stable. IFE with $s^* > M$ is locally stable while IFE with $s^* < M$ is unstable. If $y(0) < 1-M(R_0M)^{-R_0/R_1}$, then the dynamics reach a point at the IFE as $t \to \infty$; otherwise it reaches the EE. 
\end{enumerate}
\end{theorem}

We now introduce a population game model to capture how individuals adopt protective behavior, followed by analyzing the resulting evolution of the SIRI epidemic. 


\subsection{Population Game Model of Protection Adoption}

We build upon the formulation in the previous section. Consider a large population of individuals, each in one of three possible infection states, who choose whether to adopt protection or remain unprotected. The {\it population state} $x := [x_{\SUt} \quad x_{\SPt} \quad x_{\IUt} \quad x_{\IPt} \quad x_{\RUt} \quad x_{\RPt}]^\intercal \in \Delta_6$, where $x_{\RUt}$ denotes the proportion of recovered individuals who remain unprotected, $x_{\RPt}$ denotes the proportion of recovered individuals who adopt protection, and the other states are as described earlier. At time $t$, we have $x_{\SUt}(t) + x_{\SPt}(t) = s(t)$, $x_{\IUt}(t) + x_{\IPt}(t) = y(t)$, $x_{\RUt}(t) + x_{\RPt}(t) = r(t)$ and $\mathds{1}^\intercal x(t) = 1$. 

A susceptible individual becomes infected at a rate $\betap\ge0$ ($\betau\ge0$) if it comes in contact with an infected individual who adopts protection (remains unprotected). Similarly, a recovered individual becomes reinfected at rate $\hatbetau\ge0$ ($\hatbetap\ge0$) if it comes in contact with an infected and unprotected (protected) individual. A susceptible or recovered individual is $\alpha \in (0,1)$ times less likely to become infected (or reinfected) if it adopts protection compared to an unprotected individual with same disease status. Building upon the discussion in the previous section, we define the payoff vector as
\begin{align}\label{eq:payoff_siri}
    F(x) & \!:=
    \begin{bmatrix}
        F_{\SUt}(x)\\
        F_{\SPt}(x)\\
        F_{\IUt}(x)\\
        F_{\IPt}(x)\\
        F_{\RUt}(x)\\
        F_{\RPt}(x)
    \end{bmatrix}\!=\!
    \begin{bmatrix}
        -L(\betau x_{\IUt} + \betap x_{\IPt})\\
        -\cp \!-\! L\alpha(\betau x_{\IUt} + \betap x_{\IPt})\\
        -c_{\IUt}\\
        -c_{\IPt}\\
        -L(\hatbetau x_{\IUt} + \hatbetap x_{\IPt})\\
        - \cp\!-\!L\alpha(\hatbetau x_{\IUt} + \hatbetap x_{\IPt})
    \end{bmatrix},
\end{align}
where $L, \cp, c_{\IUt}$ and $c_{\IPt}$ are as defined earlier. 

The payoffs for susceptible and infected individuals defined above coincide with the payoffs in case of the SIS epidemic stated in \eqref{eq:payoff_sis}. Due to risk of reinfection, recovered individuals behave in a similar manner as susceptible individuals and evaluate the trade-off between cost of adopting protection and the instantaneous infection risk while choosing their protection status. As before, we assume $c_{\IUt} > c_{\IPt} \ge 0$. 

We denote the time-varying proportions of susceptible, infected and recovered individuals who choose to remain unprotected by $z_S$, $z_I$, and $z_R$, respectively. In particular, we have $x_{\RUt} = z_R r$, $x_{\RPt} = (1 - z_R)r$ and so on. These proportions evolve according to the replicator dynamics with payoffs defined in \eqref{eq:payoff_siri}. We now state the coupled evolution of disease and evolutionary behavioral dynamics as follows: 
\begin{subequations}\label{eq:siri_coupled} 
\begin{align}
\dot{s} & = -(\betau z_I + \betap(1 - z_I))(z_S + \alpha(1 - z_S))s y \label{eq:sdot_siri_coupled} 
\\ \dot{y} & = (\betau z_I + \betap(1 - z_I))(z_S + \alpha(1 - z_S))s y \nonumber
\\ & \qquad +\! (\hatbetau z_I + \hatbetap(1\!-\! z_I))(z_R + \alpha(1 -\! z_R)) r y \!-\! \gamma y \label{eq:ydot_siri_coupled}
\\ \dot{r} & = -\!(\hatbetau z_I + \hatbetap(1 \!-\! z_I))(z_R + \alpha(1 - z_R)) r y\!+ \gamma y  \label{eq:rdot_siri_coupled}
\\ \dot{z}_S & = z_S(1\! -\! z_S)(\cp -\! L(1 -\! \alpha)(\betau z_I + \betap(1 - z_I))y) \label{eq:zsdot_siri_coupled}
\\ \dot{z}_I & = z_I(1\! -\! z_I)(c_{\IPt} - c_{\IUt})\label{eq:zidot_siri_coupled}
\\ \dot{z}_R & = z_R(1\! -\! z_R)(\cp \!-\! L(1 \!-\! \alpha)(\hatbetau z_I \!+\! \hatbetap(1\! -\! z_I))y)\label{eq:zrdot_siri_coupled}.
\end{align}
\end{subequations}

\begin{lemma}[Invariant Set for Coupled SIRI Dynamics]\label{lem:invariant-set-siri}
For the coupled SIRI epidemic and evolutionary behavioral dynamics defined by \eqref{eq:siri_coupled}, the set $\{(s, y, r, z_S, z_I, z_R) \rvert (s, y, r, z_S, z_I, z_R) \in [0,1]^6\}$ is invariant. 
\end{lemma}

The proof follows from identical arguments as the proof of Lemma \ref{lem:invariant-set} and is omitted due to space constraints. 

A complete characterization of the equilibria and their stability properties for the above dynamics is prohibitive due to the dynamics being high dimensional and the presence of a large number of equilibrium points. In order to gain insights into epidemic evolution and convergence of infected proportion under game-theoretic decision-making, we analyze the coupled dynamics via timescale separation arguments. Specifically, we analyze in detail the case when the behavioral dynamics is much faster than the disease dynamics, derive the stable equilibria of the fast system (consisting of variables $z_S, z_I, z_R$) followed by analyzing the equilibria and convergence behavior of the reduced dynamics for the slow system (consisting of variables $s, y, r$). 

\subsection{Equilibria of the Replicator Dynamics Evolving on the Faster Timescale}

At a given infection prevalence $y$, we now characterize the equilibria of the behavioral dynamics \eqref{eq:zsdot_siri_coupled}, \eqref{eq:zidot_siri_coupled}, \eqref{eq:zrdot_siri_coupled}. 

\begin{proposition}\label{prop:siri_fast}
Let $y^*_{\mr{int}}:= \frac{c_P}{L(1-\alpha)\beta_p}$ and $\hat{y}^*_{\mr{int}}:= \frac{c_P}{L(1-\alpha)\hat{\beta}_p}$. Then, the following are true.
\begin{itemize}
\item Any equilibrium with $z_I = 1$ is unstable, 
\item For the dynamics \eqref{eq:zsdot_siri_coupled} and \eqref{eq:zidot_siri_coupled},
\begin{enumerate}
\item $(z_S = 1,z_I = 0)$ is the stable equilibrium when $y < y^*_{\mr{int}}$, 
\item $(z_S = 0,z_I = 0)$ is the stable equilibrium when $y > y^*_{\mr{int}}$,
\item any $z_S \in [0,1]$, $z_I = 0$ is an equilibrium when $y = y^*_{\mr{int}}$. 
\end{enumerate}
\item For the dynamics \eqref{eq:zidot_siri_coupled} and \eqref{eq:zrdot_siri_coupled},
\begin{enumerate}
\item $(z_I = 0,z_R = 1)$ is the stable equilibrium when $y \!< \hat{y}^*_{\mr{int}}$, 
\item $(z_I = 0,z_R = 0)$ is the stable equilibrium when $y \!> \hat{y}^*_{\mr{int}}$,
\item any $z_R \in [0,1]$, $z_I = 0$ is an equilibrium when $y = \hat{y}^*_{\mr{int}}$. 
\end{enumerate}
\end{itemize}
\end{proposition} 

The proof is straightforward and is omitted. In fact, the stable equilibria of the above behavioral dynamics corresponds to the Nash equilibrium strategies of the underlying population game defined by the payoff vector \eqref{eq:payoff_siri} at a given infection prevalence $y$. In particular, when $y = y^*_{\mr{int}}$, susceptible individuals are indifferent between adopting protection or remaining unprotected. When $y > y^*_{\mr{int}}$, it is optimal to adopt protection which leads to $z_S = 0$ being the stable equilibrium, and when $y < y^*_{\mr{int}}$, it is optimal to remain unprotected which leads to $z_S = 1$ being stable. Stability of $z_R$ is analogous. 

We now analyze the epidemic dynamics under timescale separation when the population instantly converges to the Nash equilibrium strategies depending on the current value of $y$. We consider two cases: $\betap > \hatbetap$ and $\betap < \hatbetap$ separately. 


\subsection{Epidemic Evolution under Strengthened Immunity upon Infection}

We first consider the case when $\betap > \hatbetap$, i.e., compared to a recovered individual, a susceptible individual becomes infected at a higher rate from a protected infected individual. In other words, initial infection strengthens the immunity of the individual against subsequent infections. In this case, we have $y^*_{\mr{int}} < \hat{y}^*_{\mr{int}}$. The reduced epidemic dynamics for infected and recovered subpopulations, where $z_S$ and $z_R$ are replaced by their stable equilibrium values in accordance with Proposition \ref{prop:siri_fast}, is a hybrid system given by
\begin{subequations}\label{eq:siri_strong_switched}
\begin{align}
& y \in [0,y^*_{\mr{int}}) : 
\begin{cases}\label{eq:siri_strong_lowy}
& \!\!\! \dot{y} = [\betap (1-r-y) + \hatbetap r - \gamma] y \\
& \!\!\! \dot{r} = [-\hatbetap r + \gamma] y,  
\end{cases} 
\\ 
& y = y^*_{\mr{int}} :  
\begin{cases}\label{eq:siri_strong_th1}
& \!\!\! \dot{y} \in \{  [(z_S + \alpha(1-z_S) \betap (1-r-y) + \hatbetap r - \gamma] y \ \rvert \ z_S \in [0,1] \} \\
& \!\!\! \dot{r} = [-\hatbetap r + \gamma] y, 
\end{cases} 
\\ 
& y \in (y^*_{\mr{int}},\hat{y}^*_{\mr{int}}) :  
\begin{cases}\label{eq:siri_strong_mediumy}
& \!\!\!\!\! \dot{y} = \![\alpha \betap (1\!-\!r-\!y) \!+\! \hatbetap r \!- \! \gamma] y \\
& \!\!\!\!\! \dot{r} = \![-\hatbetap r + \gamma] y, 
\end{cases} 
\\ 
& y = \hat{y}^*_{\mr{int}} : 
\begin{cases}\label{eq:siri_strong_th2}
& \!\!\!\! \dot{y} \in \{ [\alpha \betap (1-r-y) + \hatbetap r (z_R +\! \alpha(1-\!z_R)) -\! \gamma] y \rvert\! \ z_R \in [0,1] \} \\
& \!\!\!\! \dot{r} \in \{[-\hatbetap (z_R + \alpha(1-z_R) r + \gamma] y \ \rvert \ z_R \in [0,1] \},
\end{cases} 
\\ 
& y \in (\hat{y}^*_{\mr{int}},1] : 
\begin{cases}\label{eq:siri_strong_highy} 
& \!\!\!\!\! \dot{y} =\! [\alpha \betap (1-\!r-\!y) \!+ \alpha \hatbetap r -\! \gamma] y \\
& \!\!\!\!\! \dot{r} =\! [-\alpha \hatbetap r + \gamma] y. 
\end{cases} 
\end{align}
\end{subequations}

In particular, when $y < y^*_{\mr{int}} < \hat{y}^*_{\mr{int}}$, it follows from the previous subsection that the stable equilibrium of the fast system (behavioral replicator dynamics \eqref{eq:zsdot_siri_coupled}, \eqref{eq:zidot_siri_coupled}, \eqref{eq:zrdot_siri_coupled}) is $z_S = 1, z_I = 0, z_R = 1$. Setting these values in \eqref{eq:ydot_siri_coupled} and \eqref{eq:rdot_siri_coupled} together with the observation $s = 1-y-r$ yields the disease dynamics stated in \eqref{eq:siri_strong_lowy}. The dynamics for moderate and high prevalence of infection stated in \eqref{eq:siri_strong_mediumy} and \eqref{eq:siri_strong_highy} are obtained in an analogous manner. The dynamics at points of discontinuities are defined via differential inclusions. As before, it is easy to see that \eqref{eq:siri_strong_switched} admits a Filippov solution. 

We now characterize the equilibria $(y^*,r^*)$ of the dynamics in \eqref{eq:siri_strong_switched}. The following three types of equilibria are possible:
\begin{align*}
\mathbf{E1} & = (0,r^*), \quad  \text{for} \quad r^* \in [0,1], 
\\ \mathbf{E2} & = \left(1 - \frac{\gamma}{\hatbetap},\frac{\gamma}{\hatbetap}\right), \qquad  \mathbf{E3} = \left(1 - \frac{\gamma}{\alpha\hatbetap},\frac{\gamma}{\alpha\hatbetap}\right). 
\end{align*}
$\mathbf{E1}$ corresponds to a continuum of \textit{infection free equilibria (IFE)} where the proportion of infected population is $0$ while the proportion of recovered population depends on the initial condition and the values of the other parameters. This set of equilibria always exists for the dynamics in \eqref{eq:siri_strong_switched}. $\mathbf{E2}$ and $\mathbf{E3}$ are {\it endemic equilibrium} points with a nonzero proportion of infected population and the susceptible proportion being $0$. The existence and stability of all these equilibria as well as a sliding mode solution are established below. 

\medskip

\begin{proposition}[Equilibria and Stability of \eqref{eq:siri_strong_switched}]\label{prop:equilibria-and-stability-strong-siri}
	For the SIRI epidemic under game-theoretic protection stated in \eqref{eq:siri_strong_switched}, the following statements hold for any $y(0) \in (0,1]$.
	\begin{enumerate}
		\item If $\gamma > \betap$, $\mathbf{E2}$ and $\mathbf{E3}$ are not equilibrium points, the set of the IFE is globally asymptotically stable, and $y(t)$ decays monotonically to $0$. 
		\item If $\hatbetap< \gamma < \betap$, $\mathbf{E2}$ and $\mathbf{E3}$ are not equilibrium points, the sets of IFE with $r^* \lessgtr \frac{\betap - \gamma}{\betap - \hatbetap}$ are unstable and globally asymptotically stable, respectively.   
		\item If $\hatbetap [1-\hat{y}^*_{\mr{int}}]< \gamma < \hatbetap$, $\mathbf{E3}$ is not an equilibrium point, the set of the IFE is unstable, and $\mathbf{E2}$ is an (almost) globally asymptotically stable equilibrium point.
		\item If $\alpha \hatbetap [1-\hat{y}^*_{\mr{int}}] < \gamma < \hatbetap [1-\hat{y}^*_{\mr{int}}]$, then $\mathbf{E2}$ and $\mathbf{E3}$ are not equilibrium points, the set of the IFE is unstable and $(\hat{y}^*_{\mr{int}},1-\hat{y}^*_{\mr{int}})$ acts as a sliding mode of \eqref{eq:siri_strong_switched} with $y(t) \to \hat{y}^*_{\mr{int}}$ as $t \to \infty$.
		\item If $\gamma < \alpha \hatbetap [1-\hat{y}^*_{\mr{int}}]$, $\mathbf{E2}$ is not an equilibrium point, $\mathbf{E3}$ is an (almost) globally asymptotically stable equilibrium point, and the set of the IFE is unstable. 
	\end{enumerate}
\end{proposition}

The proof is presented in Appendix \ref{appendix:siri} and exploits the notion of input-to-state stability \cite{khalil}. Note that in a certain parameter regime (Case 4 in the above proposition), the infection free equilibria are not stable, endemic equilibria do not exist and infected fraction converges to a sliding mode of the hybrid dynamics. This sliding mode corresponds to an equilibrium of the original coupled dynamics \eqref{eq:siri_coupled} given by $(s = 0, y = \hat{y}^*_{\mr{int}}, r = 1 - \hat{y}^*_{\mr{int}}, z_S = 0, z_I = 0, z_R = \frac{1}{1-\alpha} \left[ \frac{\gamma}{\hatbetap(1-\hat{y}^*_{\mr{int}})} - \alpha \right])$. It is easy to see that for the range of $\gamma$ in Case 4,  $z_R \in (0,1)$, i.e., the outcome is an intermediate level of protection adoption by recovered individuals. 


\subsection{Epidemic Evolution under Compromised Immunity}

We now consider the case when $\betap < \hatbetap$, i.e., initial infection leads to compromised immunity against future infections. Here, we have $y^*_{\mr{int}} > \hat{y}^*_{\mr{int}}$. The reduced epidemic dynamics for infected and recovered subpopulations is given by 
\begin{subequations}\label{eq:siri_weak_switched}
\begin{align}
& y \in [0,\hat{y}^*_{\mr{int}}) :
\begin{cases}
& \!\!\! \dot{y} = [\betap (1-r-y) + \hatbetap r - \gamma] y \\
& \!\!\! \dot{r} = [-\hatbetap r + \gamma] y, 
\end{cases} \label{eq:siri_weak_lowy}
\\ 
& y \in (\hat{y}^*_{\mr{int}},y^*_{\mr{int}}) \!:\! 
\begin{cases}
& \!\!\!\!\! \dot{y} = [\betap (1-\!r-\!y) +\! \alpha \hatbetap r -\! \gamma] y \\
& \!\!\!\!\! \dot{r} = [-\alpha \hatbetap r + \gamma] y, 
\end{cases} \label{eq:siri_weak_mediumy}
\\ 
& y \in (y^*_{\mr{int}},1] :\!\!
\begin{cases}\label{eq:siri_weak_highy}
& \!\!\!\! \dot{y} =\! [\alpha \betap (1-r-y) +\! \alpha \hatbetap r - \!\gamma] y \\
& \!\!\!\! \dot{r} =\! [-\!\alpha \hatbetap r +\! \gamma] y, 
\end{cases} 
\end{align}
\end{subequations}
and the dynamics at $y = y^*_{\mr{int}}$ and $y = \hat{y}^*_{\mr{int}}$ can be written in terms of differential inclusions similar to \eqref{eq:siri_strong_switched}. The above hybrid dynamics is obtained by setting $z_S$ and $z_R$ values to their stable equilibrium values in accordance with Proposition \ref{prop:siri_fast}. The equilibria of \eqref{eq:siri_weak_switched} coincide with those of \eqref{eq:siri_strong_switched}, i.e., the following equilibria exist:
\begin{align*}
\mathbf{E1} & = (0,r^*), \quad  \text{for} \quad r^* \in [0,1], 
\\ \mathbf{E2} & = \left(1 - \frac{\gamma}{\hatbetap},\frac{\gamma}{\hatbetap}\right), \qquad  \mathbf{E3} = \left(1 - \frac{\gamma}{\alpha\hatbetap},\frac{\gamma}{\alpha\hatbetap}\right). 
\end{align*}
The existence and local stability of all these equilibria as well as a sliding mode solution are established below with proof presented in Appendix \ref{app:sec:last}. 

\smallskip

\begin{proposition}[Equilibria and Stability of \eqref{eq:siri_weak_switched}]\label{prop:equilibria-and-stability-weak-siri}
For the SIRI epidemic under game-theoretic protection stated in \eqref{eq:siri_weak_switched}, the following statements hold:
\begin{enumerate}
\item if $\gamma > \hatbetap$, $\mathbf{E2}$ and $\mathbf{E3}$ are not equilibrium points, the set of the IFE is globally asymptotically stable, and $y(t)$ decays monotonically to $0$;
\item if $\hatbetap [1-\hat{y}^*_{\mr{int}}]< \gamma < \hatbetap$, $\mathbf{E2}$ is an equilibrium point which is locally asymptotically stable, and $\mathbf{E3}$ is not an equilibrium point;
\item if $\alpha \hatbetap [1-\hat{y}^*_{\mr{int}}] < \gamma < \hatbetap [1-\hat{y}^*_{\mr{int}}]$, then $\mathbf{E2}$ and $\mathbf{E3}$ are not equilibrium points, and $(\hat{y}^*_{\mr{int}},1-\hat{y}^*_{\mr{int}})$ acts as a sliding mode of \eqref{eq:siri_weak_switched};
\item if $\gamma < \alpha \hatbetap [1-\hat{y}^*_{\mr{int}}]$, $\mathbf{E2}$ is not an equilibrium point, $\mathbf{E3}$ is a locally asymptotically stable equilibrium point.  
\end{enumerate}
In addition, for statements $2-4$ above, 
\begin{itemize}
\item if $\gamma < \betap$, then all points in the IFE $\mathbf{E1}$ are unstable;
\item if $\gamma > \betap$, then the sets of the IFE  with $r^* \lessgtr \frac{\betap - \gamma}{\betap - \hatbetap}$ are locally stable and unstable, respectively. 
\end{itemize} 
\end{proposition}

\begin{remark}
Note that when $\gamma > \betap$ and $\gamma$ satisfies any of the conditions in statements $2-4$ of Proposition \ref{prop:equilibria-and-stability-weak-siri}, the hybrid dynamics exhibits {\it bistability} where both the IFE and an endemic equilibrium or attractive sliding mode coexist and are stable. In contrast, when $\betap > \hatbetap$, multiple stable equilibria do not coexist. Due to the coexistence of multiple stable equilibria, characterizing the respective regions of attraction remains a direction for future research. 
\end{remark}
\section{Numerical Results}
\label{section:numerical}

We now provide further insights into the coupled epidemic-replicator dynamics via numerical simulations. 

\subsection{SIS Epidemic Setting}

\begin{figure}[t!]
	\centering
	\includegraphics[width=0.8\linewidth]{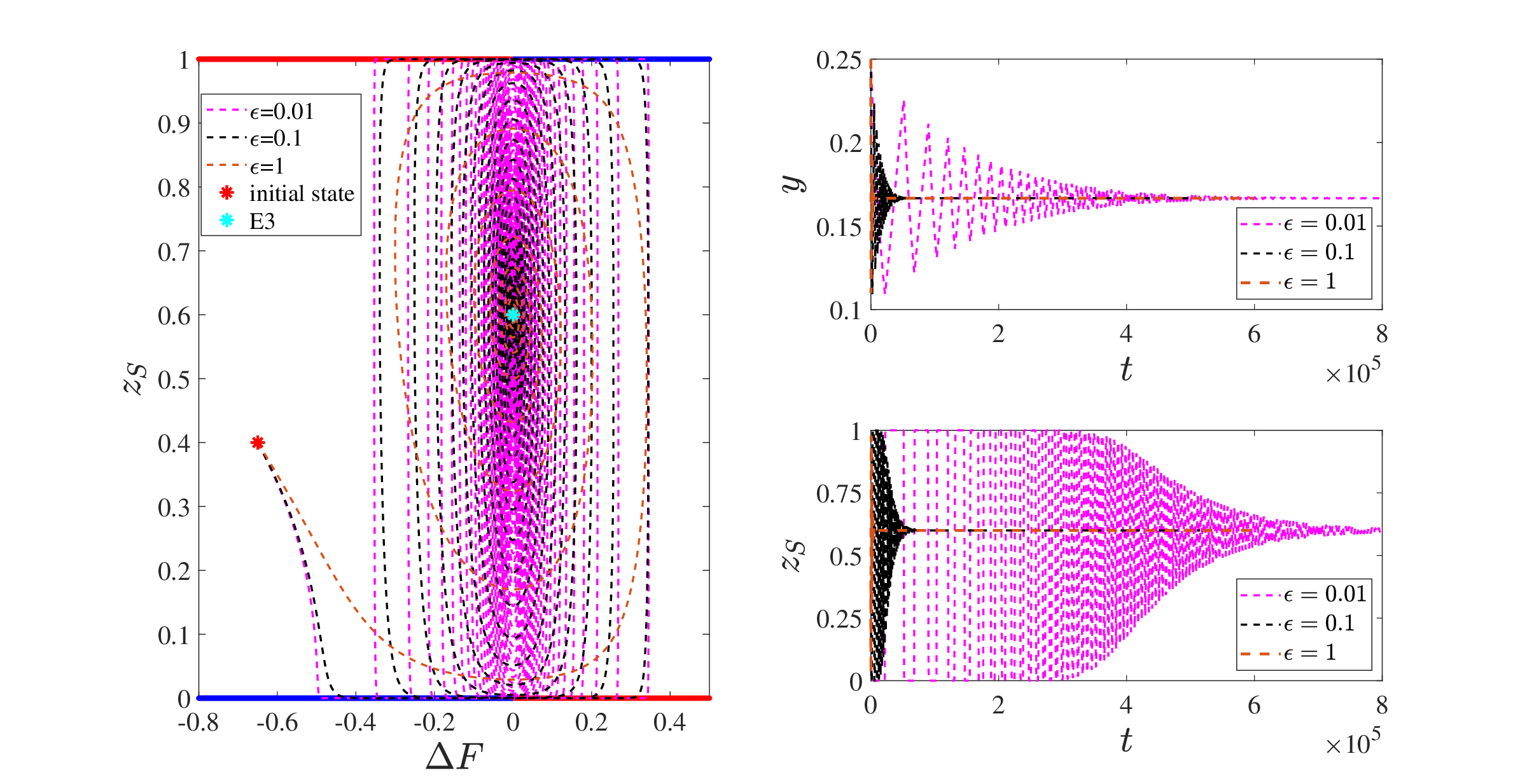}
	\caption{Trajectories of the coupled SIS epidemic-behavioral dynamics when the behavioral dynamics evolves on a faster timescale. The right panel shows the oscillatory trajectories of infected proportion ($y$) and susceptible proportion that does not adopt protection ($z_S$). The oscillations increase as $\epsilon$ is decreased. The left panel shows trajectories in $z_S-\Delta F$ plane. Trajectories reach an equilibrium point at which $\Delta F =0$.}
	\label{fig:slow_fast_ga_p1}
\end{figure}

\begin{figure*}[t]
\centering
    \includegraphics[scale=0.32]{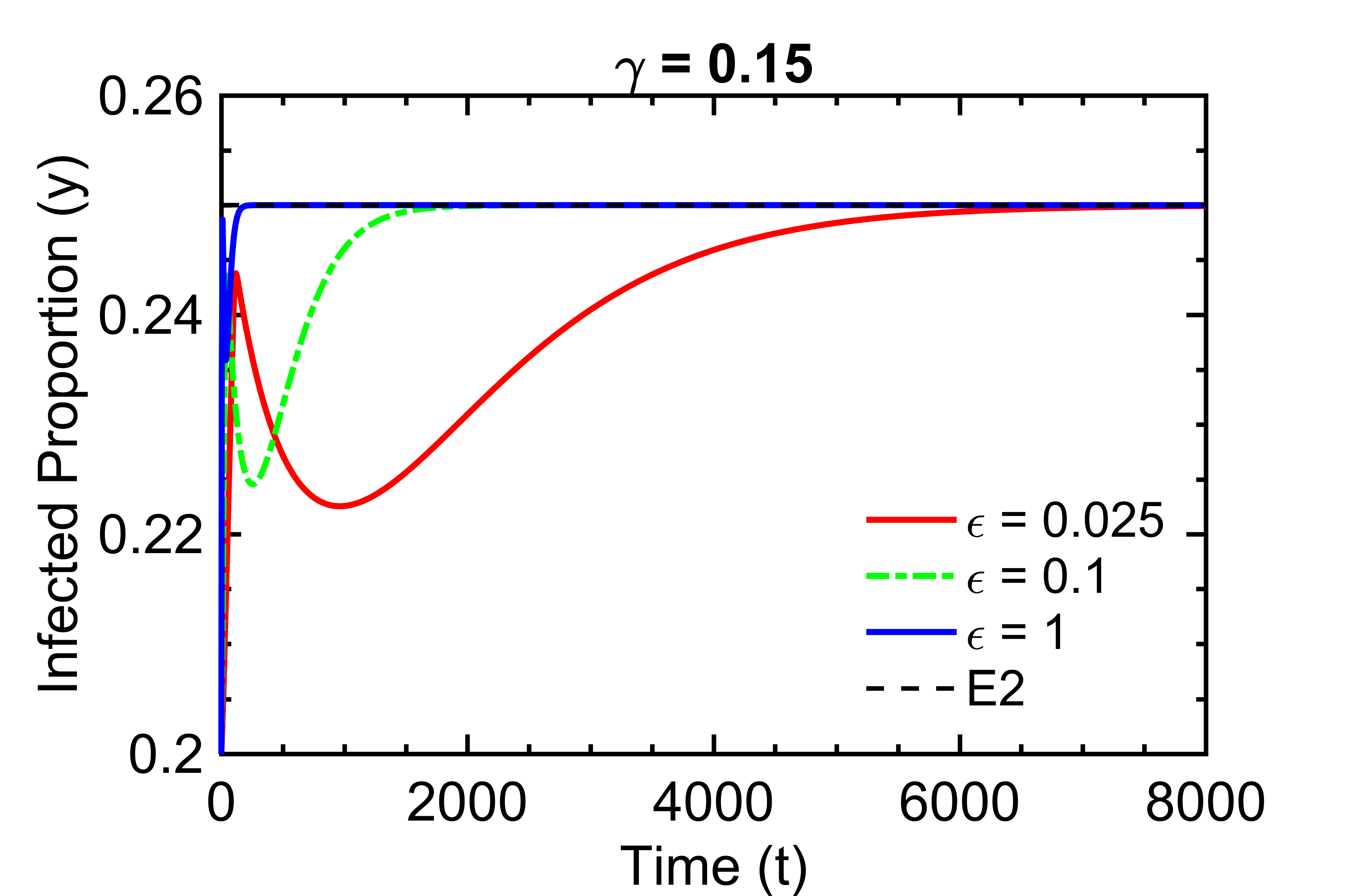}
    ~\includegraphics[scale=0.32]{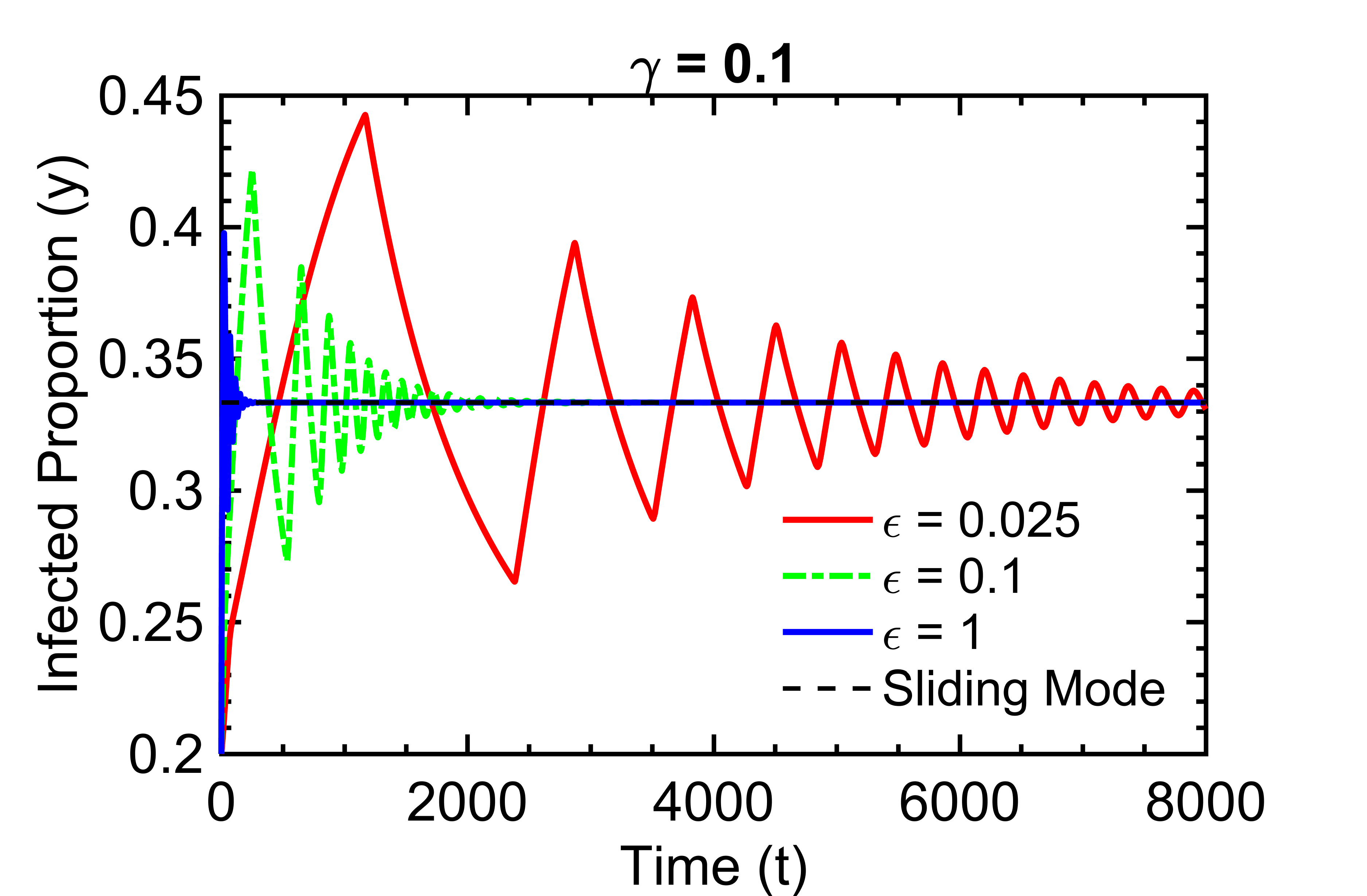}
    ~\includegraphics[scale=0.32]{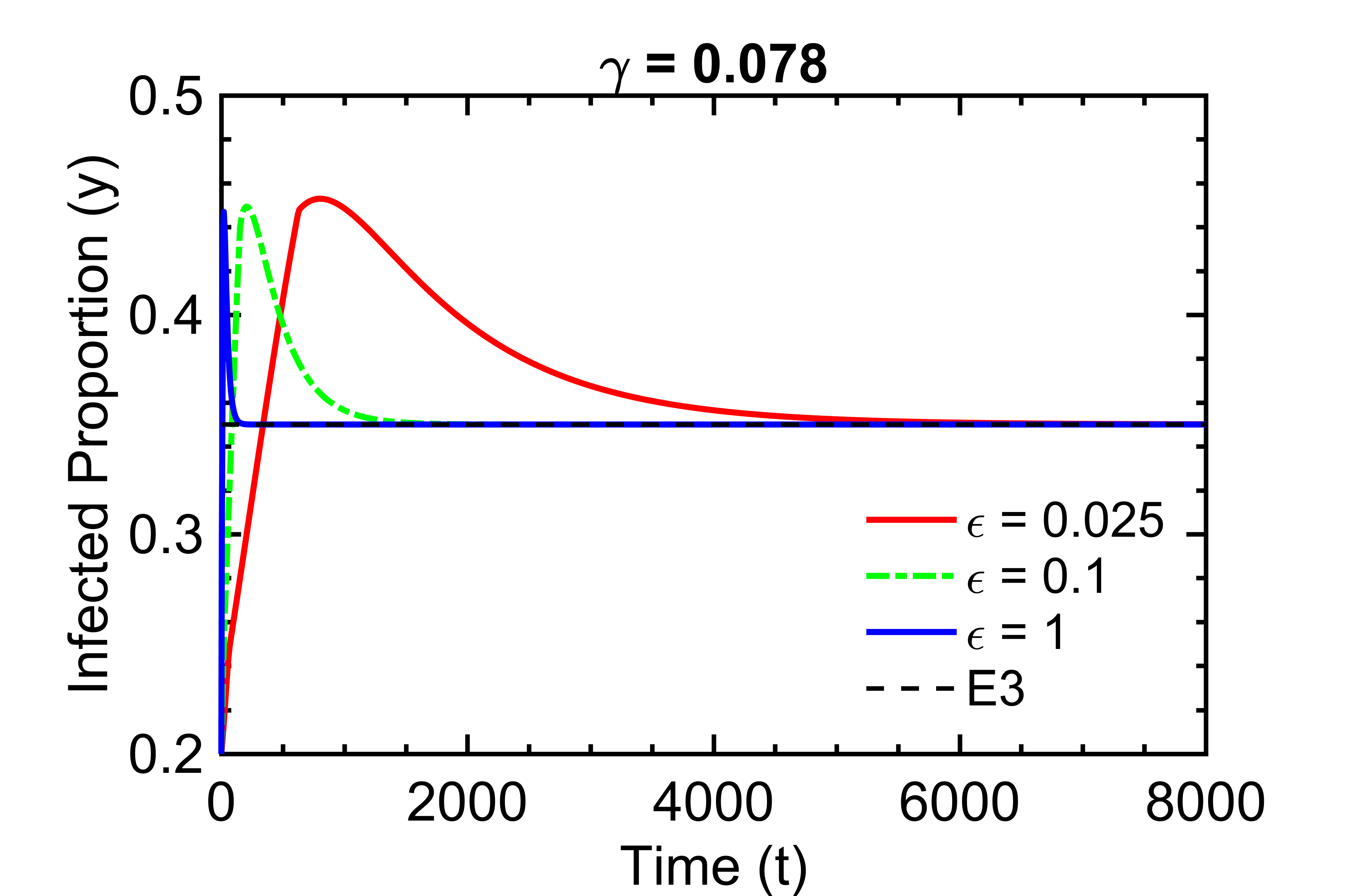}
    \caption{\small Evolution of infected proportion in the SIRI model for different values of recovery rate $\gamma$ and timescale separation parameter $\epsilon$ in accordance with Proposition \ref{prop:equilibria-and-stability-strong-siri} when initial infection leads to strengthened immunity.}
    \label{fig:SIRI_strong}
\end{figure*}

We first investigate the SIS epidemic setting. We simulate the dynamics in \eqref{eq:coupled-dynamics-timescale_slow_epi} for $\epsilon \in \{0.01, 0.1, 1\}$.  The same model parameters as in Section~\ref{sec:bif-analysis} are selected and $\gamma$ is selected as $0.1$ so that $\mathbf{E3}$ is the stable equilibrium point. Fig.~\ref{fig:slow_fast_ga_p1} (right panel) shows the time-evolution of $y$ and $z_S$. More oscillatory behavior is observed as $\epsilon$ becomes smaller, i.e., the behavioral dynamics becomes faster. To understand this, we focus on $z_S$ dynamics~\eqref{eq:main_zs} with 
	\[\Delta F =   \cp - L(1-\alpha)(\betau z_{I}(t)  + \betap (1-z_I(t))) y(t)
\]
as a dynamic parameter. To this end, we illustrate $z_S$ trajectories in $z_S-\Delta F$ plane in  Fig.~\ref{fig:slow_fast_ga_p1} (left panel). Recall that if $\Delta F$ is a positive (resp. negative) constant, then $z_S =1$ (resp. $z_S=0$) is a stable equilibrium point. Accordingly, $z_S  =0$ and $z_S=1$ are marked blue and red in Fig.~\ref{fig:slow_fast_ga_p1} (left panel), when they are stable and unstable, respectively. 

Since behavioral dynamics is fast, $y$ is quasi-stationary and $z_I$ very quickly converges to zero. In Fig.~\ref{fig:slow_fast_ga_p1} (left panel), the initial fraction of infected population is sufficiently high such that $\Delta F <0$, then the fast behavioral dynamics quickly converges to $z_S =0$ (the bottom solid blue line), i.e., every susceptible individual adopts protection. This results in a decrease in the fraction of infected population and increases $\Delta F$. As $\Delta F$ becomes positive, $z_S=0$ becomes unstable and $z_S$ quickly jumps to $z_S =1$ (the top solid blue line), and a similar process repeats which again drives $\Delta F$ to negative values. This process leads to the highly oscillatory behavior seen in Fig.~\ref{fig:slow_fast_ga_p1}. Eventually, trajectories converge such that $\Delta F =0$ and $z_S$ settles to equilibrium value $z^*_{S,\mr{int}}$. 

\subsection{SIRI Epidemic Setting}

\begin{figure*}
\centering
    \includegraphics[scale=0.32]{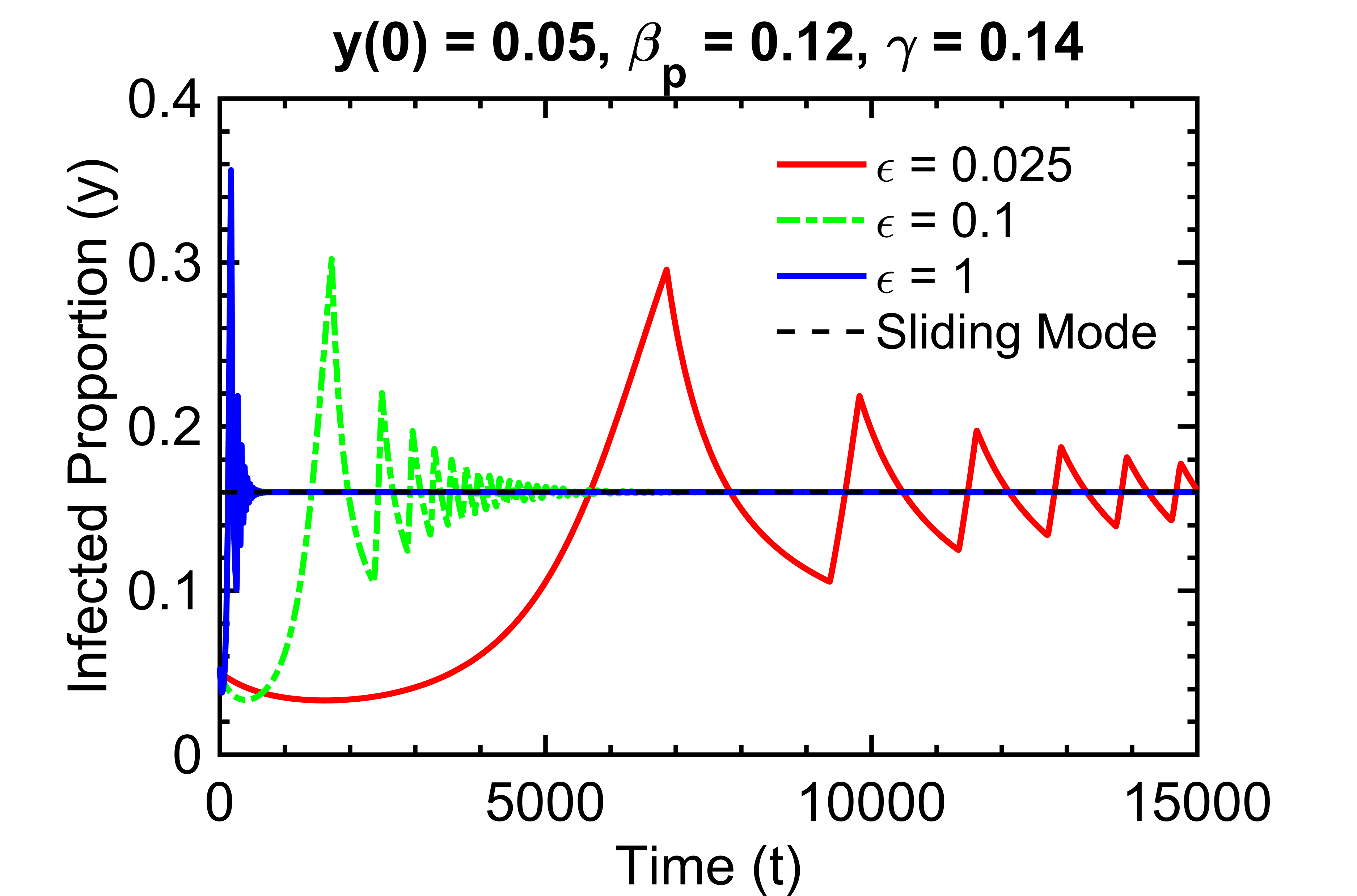}
    ~\includegraphics[scale=0.32]{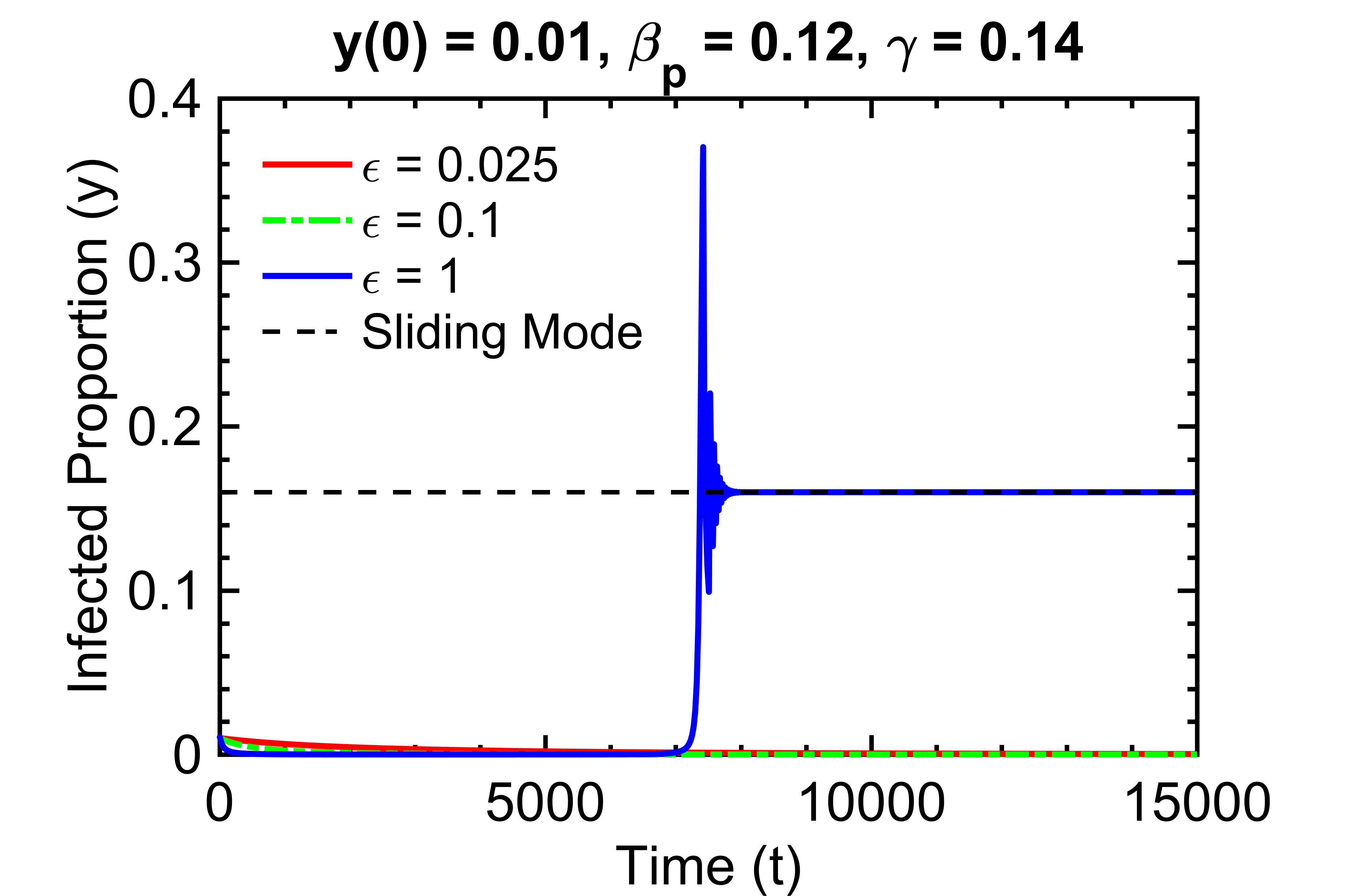}
    ~\includegraphics[scale=0.32]{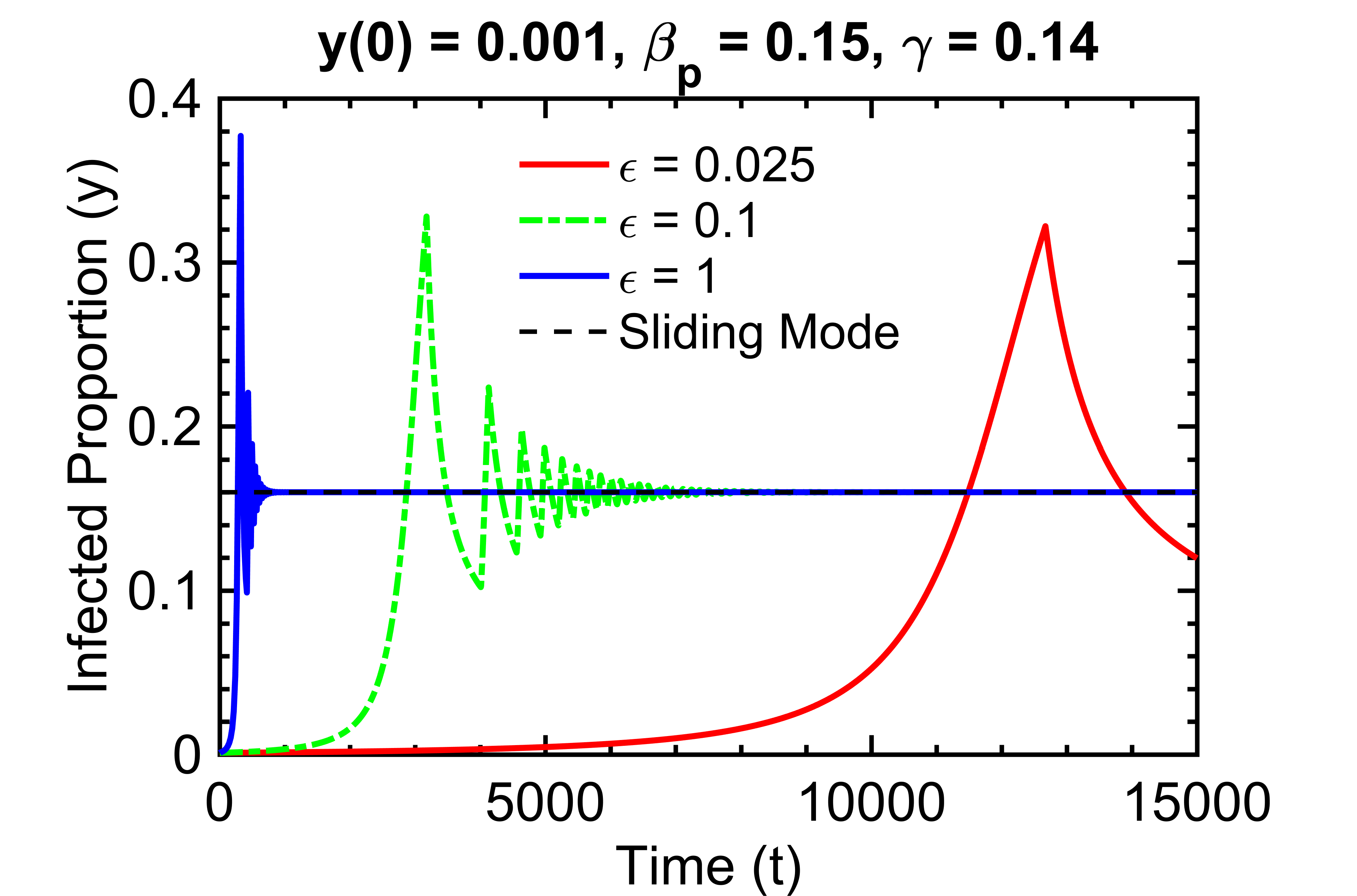}
    \caption{\small Evolution of infected proportion and bistable behavior in the SIRI model for different values of $\gamma$ and $\epsilon$ when initial infection leads to compromised immunity in accordance with Proposition \ref{prop:equilibria-and-stability-weak-siri}.}
    \label{fig:SIRI_bistable}
\end{figure*}

We first consider the case where initial infection leads to strengthened immunity. In order to highlight a wide range of transient behavior of the coupled dynamics, we choose parameter values as $\betap = 0.3, \hatbetap = 0.2, \betau = 0.4,  \hatbetau = 0.25, L = 75, \alpha = 0.6, \cp = 2, c_{\IUt} = 2, c_{\IPt} = 1$.\footnote{The value of $L$ is chosen to be much larger than the cost of protection to reflect the fact that the loss upon infection is much more significant both in terms of health risks as well as the economic loss (due to being quarantined or even hospitalized for many days) compared to the cost of wearing masks that are essentially free.} For this set of parameter values, we have $y^*_{\mr{int}} = 0.2222$ and $\hat{y}^*_{\mr{int}} = 0.3333$. Fig. \ref{fig:SIRI_strong} shows the evolution of infected proportion $y(t)$ for three different values of the recovery rate $\gamma$ with initial states as $y(0) = 0.2, r(0) = 0.01, z_S(0) = z_I(0) = z_R(0) = 0.5$. In particular, we compute the trajectories of the dynamics in \eqref{eq:siri_coupled} with $\epsilon$ multiplied to the R.H.S. of \eqref{eq:sdot_siri_coupled}, \eqref{eq:ydot_siri_coupled}, \eqref{eq:rdot_siri_coupled} via Runge-Kutta fourth order method with spacing $0.05$. When $\epsilon = 1$, epidemic and replicator dynamics evolve at the same timescale while a smaller value of $\epsilon$ signifies a faster evolution of behavior compared to disease evolution. 

When $\gamma = 0.15 \in (\hatbetap [1-\hat{y}^*_{\mr{int}}], \hatbetap)$, it follows from Proposition \ref{prop:equilibria-and-stability-strong-siri} that points at the infection free equilibria are not stable, $\mathbf{E3}$ does not exist while $\mathbf{E2}$ is a locally stable equilibrium. The plot in the left panel of Fig. \ref{fig:SIRI_strong} shows indeed $y(t)$ converges to the endemic infection level at $\mathbf{E2}$ in our simulations. Furthermore, the convergence is not monotonic. When $\gamma = 0.1$, $y(t)$ converges to the endemic infection level $\hat{y}^*_{\mr{int}}$ in accordance with Case $4$ of Proposition \ref{prop:equilibria-and-stability-strong-siri}. The convergence is in an oscillatory manner similar to the setting illustrated above for the SIS epidemic. Finally, when $\gamma = 0.078 < \alpha \hatbetap [1-\hat{y}^*_{\mr{int}}]$, $y(t)$ converges to the endemic infection level at $\mathbf{E3}$.

Finally, we consider the case where initial infection leads to compromised immunity. We choose parameter values as $\hat{\beta}_p = 0.25, \hatbetau = 0.4,  \betau = 0.35, L = 125, \alpha = 0.6, \cp = 2, c_{\IUt} = 2, c_{\IPt} = 1$. For this set of parameter values, we have $\hat{y}^*_{\mr{int}} = 0.16$. Recall that when $\betap <  \hatbetap$ and $\gamma > \betap$, infection free equilibria as well as an endemic equilibrium point can be simultaneously locally stable. Fig. \ref{fig:SIRI_bistable} shows the evolution of infected proportion $y(t)$ for three different values $\gamma, \betap$ and $y(0)$ indicated on the titles of the respective subfigures. Initial values of other states are set as stated above. When $\gamma = 0.14$ and $\betap = 0.12$, both the sliding mode solution $\hat{y}^*_{\mr{int}}$ and the IFE are locally stable. From the plot in the left panel of Fig. \ref{fig:SIRI_bistable}, we observe that when $y(0) = 0.05$, $y(t)$ converges to the sliding mode solution for all three $\epsilon$ values. However, when $y(0) = 0.01$, $y(t)$ converges to the IFE for $\epsilon = 0.1, 0.025$ indicating that certain equilibria at the IFE are also locally stable when replicator dynamics is sufficiently fast in accordance with Proposition \ref{prop:equilibria-and-stability-weak-siri}. Further, when $\epsilon = 1$, $y(t)$ converges to $\hat{y}^*_{\mr{int}}$. Thus, our results show that the relative speed of epidemic and behavioral dynamics, captured by the timescale separation parameter, may determine which equilibrium the infected proportion would eventually converge to in the bistable regime of the SIRI epidemic.

Finally, when $\betap = 0.15 > \gamma$, we no longer have bistable behavior and $y(t)$ converges to $\hat{y}^*_{\mr{int}}$ for all values of $\epsilon$ even when $y(0) = 0.001$ as shown in the right panel of Fig. \ref{fig:SIRI_bistable}. 
\section{Conclusion}\label{sec:conclusions}

We proposed and analyzed a novel model that captures the interaction of epidemic propagation dynamics with behavioral dynamics of human protection adoption. For the coupled SIS epidemic-replicator dynamics, we characterized the equilibrium points, their stability properties, and the associated bifurcations. For both SIS and SIRI epidemic models, we further analyzed the coupled dynamics under timescale separation, and established global convergence results to the equilibrium of the reduced epidemic dynamics. Numerical results showed that the relative speed of evolutionary learning compared to disease dynamics plays a critical role in the transient behavior of the coupled dynamics, may induce highly oscillatory behavior, and in the bistable regime of the SIRI epidemic, it may have a strong influence on which equilibrium the infected proportion converges to.

Thus, our results highlight that in order to influence and control epidemic prevalence, it is critical to understand not only the equilibrium behavior of humans, but also the transient evolution of human behavior in a comparable time-scale as the disease dynamics. We plan to build upon the results derived in this paper, and design dynamic intervention schemes (for example, dynamically varying the cost of protection) that guarantee convergence of the disease dynamics to desired equilibrium points. Similarly, it would be interesting to consider a setting where the the payoff of infected individuals also depends on the proportion of infected individuals, for instance due to social or peer influence. Another possible factor is bounded rationality of agents that may lead to a proportion of infected agents not adopting protection. We hope that our work stimulates further investigations along the above lines.

\appendix

\section{Proofs omitted from Section \ref{section:model}}    
\label{appendix:sis}

We first evaluate the entries of the Jacobian matrix of the coupled dynamics \eqref{eq:sis_scalar_com}, \eqref{eq:main_zs} and \eqref{eq:main_zi} required to determine the stability of the equilibrium points. Specifically, we compute
\begin{align*}
\frac{\partial f_y}{\partial y} & = \big[(1-\!y) (z_S +\! \alpha (1-\! z_S)) (\betau z_I \!+ \betap (1-\!z_I))\! -\! \gamma \big] 
\\ & \qquad - y (z_S + \alpha(1-z_S)) (\betau z_I + \betap (1-z_I)), 
\\ \frac{\partial f_y}{\partial z_S} & = y \big[(1-y) (1-\alpha) (\betau z_I + \betap (1-z_I)) \big],
\\ \frac{\partial f_y}{\partial z_I} & = y \big[(1-y) (z_S + \alpha (1-z_S)) (\betau - \betap)\big].
\end{align*}
Similarly, 
\begin{align*}
\frac{\partial f_S}{\partial y} & = - z_{S}(1-{z}_{S}) L(1-\alpha)(\betau z_{I} + \betap (1-z_I)),
\\ \frac{\partial f_S}{\partial z_S} & = (1-2z_S) \big[\cp - L(1-\alpha)(\betau z_{I} + \betap (1-z_I)) y\big]
\\ \frac{\partial f_S}{\partial z_I} & = - z_{S}(1-{z}_{S}) L(1-\alpha)(\betau - \betap) y
\end{align*}
Finally,
\begin{align*}
\frac{\partial f_I}{\partial y} & = 0, \quad  \frac{\partial f_I}{\partial z_S}  = 0, \quad \frac{\partial f_I}{\partial z_I}  = (1-2z_I) (c_{\IPt}-c_{\IUt}).
\end{align*}

\subsection{Proof of Proposition \ref{prop:equilibria-and-stability}}
\label{app:sis:1}

\begin{proof}
It can be verified that $\mathbf{E0}$ and $\mathbf{E1}$  are always equilibria of the coupled dynamics. The Jacobian matrix at $\mathbf{E0}$ and $\mathbf{E1}$ are 
\begin{equation*}
	J_{\mathbf{E0}} = 
	\left[\begin{smallmatrix}
		\alpha \betap - \gamma & 0 & 0\\
		0 & \cp & 0\\
		0 & 0 & c_{\IPt} - c_{\IUt}
	\end{smallmatrix}\right], \quad J_{\mathbf{E1}} = 
\left[\begin{smallmatrix}
\betap - \gamma & 0 & 0\\
0 & -\cp & 0\\
0 & 0 & c_{\IPt} - c_{\IUt}
\end{smallmatrix}\right]. 
\end{equation*}
 $J_{\mathbf{E0}}$ has a positive eigenvalue if $\cp > 0$. Thus, for any nonzero cost of adopting protection, $\mathbf{E0}$ is not a stable equilibrium point of the coupled dynamics.\footnote{This is intuitive: in the absence of infection, it is optimal to not choose costly protective measures.} Likewise, $\mathbf{E1}$ is stable if $\betap < \gamma$, and unstable, otherwise.

We now analyze existence and local stability of equilibrium points where infection is endemic. It can be verified that $\mathbf{E2}$ exists only when $\betap > \gamma$ as otherwise $y^*_{\mathtt{U}} < 0$. The Jacobian matrix at $\mathbf{E2}$ is
\begin{equation*}
	J_{\mathbf{E2}} = 
	\left[\begin{smallmatrix}
		d_1 & d_2 & d_3\\
		0 & -[\cp - L(1-\alpha)\betap y^*_{\mathtt{U}}] & 0\\
		0 & 0 & c_{\IPt} - c_{\IUt}
	\end{smallmatrix}\right]
	,
\end{equation*}
where $d_1 = (1-2y^*_{\mathtt{U}})\betap - \gamma$, $d_2 = y^*_{\mathtt{U}}(1-y^*_{\mathtt{U}})(1-\alpha)\betap > 0$ and $d_3 = y^*_{\mathtt{U}}(1-y^*_{\mathtt{U}})(\betau - \betap) > 0$. Thus, $J_{\mathbf{E2}}$ is an upper triangular matrix. Furthermore, the first diagonal entry is
\begin{align*}
	(1-2y^*_{\mathtt{U}})\betap - \gamma &
	= (1 - 2(1-\frac{\gamma}{\betap}))\betap - \gamma 
	\\ & = (-1 + 2\frac{\gamma}{\betap})\betap - \gamma = \gamma - \betap < 0,
\end{align*}
in the regime where $\mathbf{E2}$ exists. Therefore, $\mathbf{E2}$ is stable when 
\begin{align*}
	\cp & > L(1-\alpha)\betap y^*_{\mathtt{U}}  \iff  y^*_{\mathtt{U}}  < y^*_{\mr{int}}.
\end{align*}

It can be verified that $\mathbf{E3} = (y^*_{\mr{int}},z^*_{S,\mr{int}},0)$ is an equilibrium point. We now examine the conditions under which $y^*_{\mr{int}} \in (0,1)$ and $z^*_{S,\mr{int}} \in (0,1)$.  By definition, $y^*_{\mr{int}} > 0$. We now observe that
\begin{align*}
	z^*_{S,\mr{int}} > 0 & \iff \frac{\gamma}{\alpha \betap} > 1 - y^*_{\mr{int}} \iff  y^*_{\mr{int}} > 1 - \frac{\gamma}{\alpha \betap}, \\
	z^*_{S,\mr{int}} < 1 & \iff\frac{\gamma}{\betap} < 1 - y^*_{\mr{int}} \iff  y^*_{\mr{int}} < 1 - \frac{\gamma}{\betap}. 
\end{align*}
Thus,  $\mathbf{E3}$ exists when $ y^*_{\mathtt{P}}< y^*_{\mr{int}} <  y^*_{\mathtt{U}}$.  Note that the third row of the Jacobian of the dynamics at $\mathbf{E3}$,  $J_{\mathbf{E3}}$, would be $[0 \quad 0 \quad c_{\IPt} - c_{\IUt}]$ as before, and as a result, $c_{\IPt} - c_{\IUt} < 0$ would be an eigenvalue. Thus, we focus on the $2 \times 2$ sub-matrix containing the first two rows and columns of the Jacobian matrix which simplifies to
\begin{align*}
	\hat{J}_{\mathbf{E3}} = 
	\left[\begin{smallmatrix}
		\frac{-\gamma y^*_{\mr{int}}}{1-y^*_{\mr{int}}} & d_4\\
		- z^*_{S,\mr{int}}(1-z^*_{S,\mr{int}})L(1-\alpha)\betap & 0\\
	\end{smallmatrix}\right],
\end{align*}
where $d_4 = y^*_{\mr{int}}(1-y^*_{\mr{int}})(1-\alpha)\betap > 0$. For the above matrix, the sum of the eigenvalues is negative and the determinant is positive, and as a result, $J_{\mathbf{E3}}$ is Hurwitz. Therefore, $\mathbf{E3}$, when it exists, is a stable equilibrium of the coupled dynamics. 

It can be verified that  $\mathbf{E4}$ exists when $y^*_{\mathtt{P}} \in (0,1)$ or equivalently, when $\gamma < \alpha \betap$. The Jacobian matrix at $\mathbf{E4}$ is
\begin{equation}
	J_{\mathbf{E4}} = 
	\left[\begin{smallmatrix}
		(1-2y^*_{\mathtt{P}})\alpha\betap - \gamma & d_5 & d_6\\
		0 & d_7 & 0\\
		0 & 0 & c_{\IPt} - c_{\IUt}
	\end{smallmatrix}\right],
\end{equation}
where $d_5 =  y^*_{\mathtt{P}}(1-y^*_{\mathtt{P}}) (1-\alpha)\betap> 0$, $d_6 = y^*_{\mathtt{P}}(1-y^*_{\mathtt{P}})\alpha(\betau - \betap)> 0$ and $d_7 = \cp - L(1-\alpha)\betap y^*_{\mathtt{P}}$. Thus, $J_{\mathbf{E4}}$ is an upper triangular matrix with the first diagonal entry
\begin{align*}
	(1-2y^*_{\mathtt{P}})\alpha\betap - \gamma  
	& = \gamma - \alpha\betap < 0,
\end{align*}
in the regime where $\mathbf{E4}$ exists. Thus, $\mathbf{E4}$ is stable if 
\begin{align*}
	\cp & < L(1-\alpha)\betap y^*_{\mathtt{P}}  \iff   y^*_{\mathtt{P}}  > y^*_{\mr{int}}.
\end{align*}
This concludes the proof.
\end{proof}

\subsection{Proof of Proposition \ref{prop:epidemic_slow}}
\label{app:sis:2}

\begin{proof}
Recall for the scalar SIS epidemic dynamics 
$$ \dot{y}(t) = [(1-y(t)) \beta - \gamma] y(t), $$
the infected proportion $y(t)$ monotonically converges to $0$ if $\gamma \geq \beta$ and converges to $y^* = 1-\frac{\gamma}{\beta}$, otherwise \cite{mei2017dynamics}. Furthermore, it is easy to see that in the latter case, $\rvert y(t) - y^* \rvert$ is monotonically decreasing.\footnote{For $V(y) = (y-y^*)^2$, we have $\dot{V}(y) = -2\beta yV(y) \leq 0$.}  

Note that when $y < y^*_{\mr{int}}$ (respectively, $y > y^*_{\mr{int}}$), \eqref{eq:epi_slow} is analogous to the above dynamics with $\beta = \betap$ (respectively, $\beta=\alpha \betap$). Recall further that since $\alpha \in (0,1)$, we have $y^*_{\mathtt{P}} < y^*_{\mathtt{U}}$. We now analyze the four cases stated above. 

\noindent {\bf Case 1: $y^*_{\mathtt{U}} \leq 0$}. In this case, we have $\alpha \betap < \betap \leq \gamma$ and as a result, $y = 0$ is the only equilibrium of \eqref{eq:epi_slow} both when $y > y^*_{\mr{int}}$ and $y < y^*_{\mr{int}}$. Further, $\dot{y}(t) < 0$ for any $y(t) > 0$ including at the neighborhood of $y(t) = y^*_{\mr{int}}$. 

\noindent {\bf Case 2: $0 < y^*_{\mathtt{U}} < y^*_{\mr{int}}$}. We first show that $y(t)$ is monotonically decreasing when $y(t) > y^*_{\mr{int}}$. Note that when $y(t) > y^*_{\mr{int}}$, \eqref{eq:epi_slow} resembles the scalar SIS dynamics with $\beta = \alpha \betap$. We have two possibilities.
\begin{itemize}
\item[a.] $\gamma \geq \alpha \betap$: It is easy to see that $\dot{y}(t) < 0$ for $y(t) > y^*_{\mr{int}}$ in \eqref{eq:epi_slow}, and as a result, $y(t)$ is monotonically decreasing. 
\item[b.] $\gamma < \alpha \betap$: From the above discussion, we note that $y(t)$ under the scalar SIS dynamics with $\beta = \alpha \betap$ would converge monotonically to $1 - \frac{\gamma}{\alpha \betap} = y^*_{\mathtt{P}}$. Since $\alpha < 1$, $y^*_{\mathtt{P}} < y^*_{\mathtt{U}}$ which implies $y^*_{\mathtt{P}} < y^*_{\mr{int}}$. Therefore, $y(t)$ is monotonically decreasing for $y(t) > y^*_{\mr{int}}$. 
\end{itemize}

We now focus on the dynamics when $y(t) < y^*_{\mr{int}}$ which resembles the scalar SIS dynamics with $\beta = \betap$. Since $0 < y^*_{\mathtt{U}}$, we have $\betap > \gamma$. As a result, for any $y(0) \in (0,y^*_{\mr{int}})$, $y(t)$ converges monotonically to $y^*_{\mathtt{U}}$. Since $y^*_{\mathtt{U}} < y^*_{\mr{int}}$, $\dot{y}(t) < 0$ in the neighborhood of $y^*_{\mr{int}}$. Thus, the claim follows. 

\noindent {\bf Case 3: $y^*_{\mathtt{P}} < y^*_{\mr{int}} < y^*_{\mathtt{U}}$}. Following analogous arguments as above, we observe that $y(t)$ is monotonically decreasing when $y(t) > y^*_{\mr{int}}$ as $\max(0,y^*_{\mathtt{P}}) < y^*_{\mr{int}}$, and monotonically increasing when $y(t) < y^*_{\mr{int}}$ since the corresponding stable equilibrium $y^*_{\mathtt{U}} > y^*_{\mr{int}}$. Thus, $y^*_{\mr{int}}$ acts as a sliding surface for the dynamics \eqref{eq:epi_slow} with $y(t) = y^*_{\mr{int}}$ being a Filippov solution of \eqref{eq:epi_slow} with  $0 \in \left\{ \big[(1-y^*_{\mr{int}})\betap (z_S + \alpha(1-z_S)) - \gamma \big] y^*_{\mr{int}} \rvert z_S \in [0,1]  \right\}$.

\noindent {\bf Case 4: $y^*_{\mathtt{P}} > y^*_{\mr{int}}$}. In this case, $y(t)$ is monotonically increasing when $y(t) < y^*_{\mr{int}}$ since $y^*_{\mathtt{U}} > y^*_{\mathtt{P}} > y^*_{\mr{int}}$. When $y(t) > y^*_{\mr{int}}$, the corresponding stable equilibrium is $y^*_{\mathtt{P}} \in (y^*_{\mr{int}},1)$, and $y(t)$ converges monotonically to $y^*_{\mathtt{P}}$ for $y(t) > y^*_{\mr{int}}$.
\end{proof}


\section{Proofs omitted from Section \ref{section:SIRI}} 

\subsection{Proof of Proposition \ref{prop:equilibria-and-stability-strong-siri}}
\label{appendix:siri}

\begin{proof}
We prove each of the cases below. 

\smallskip

\noindent {\bf Case 1: $\gamma > \betap$}. Since $\alpha \in (0,1)$ together with $\gamma > \betap > \hatbetap$, we have $1 - \frac{\gamma}{\alpha\hatbetap} < 1 - \frac{\gamma}{\hatbetap} < 0$. In other words, the endemic infection level is negative at $\mathbf{E2}$ and $\mathbf{E3}$. As a result, $\mathbf{E2}$ and $\mathbf{E3}$ are not equilibrium points. The local stability of the IFE depends on the dynamics \eqref{eq:siri_strong_lowy}. Drawing analogy with the SIRI dynamics in \eqref{eq:siri_vanilla}, we have $R_0 = \frac{\betap}{\gamma} < 1$ and $R_1 = \frac{\hatbetap}{\gamma} < 1$ for the dynamics \eqref{eq:siri_strong_lowy}. Since $\alpha < 1$, we have $R_0 < 1$ and $R_1 < 1$ for the dynamics \eqref{eq:siri_strong_mediumy} and \eqref{eq:siri_strong_highy} as well. As a result, following Case 1 of Theorem \ref{theorem:siri_vanilla}, the set of IFE of \eqref{eq:siri_strong_switched} is locally stable and $y(t)$ decays monotonically to $0$. 

\noindent {\bf Case 2: $\hatbetap< \gamma < \betap$}. $\mathbf{E2}$ and $\mathbf{E3}$ not being equilibrium points follows from the arguments in Case 1 above. We have $R_0 = \frac{\betap}{\gamma} > 1$ and $R_1 = \frac{\hatbetap}{\gamma} < 1$ for the dynamics \eqref{eq:siri_strong_lowy}. Following Case 3 of Theorem \ref{theorem:siri_vanilla}, the set of the IFE with $s^* < \frac{1-R_1}{R_0 - R_1}$ or equivalently with $r^* > 1 - \frac{1-R_1}{R_0 - R_1} = \frac{\betap - \gamma}{\betap - \hatbetap}$ is locally stable, and unstable otherwise.

We now argue that any trajectory of \eqref{eq:siri_strong_switched} with $y(0) \in (0,1]$ converges to a stable IFE. Recall that any IFE point $(y=0, r = \bar{r})$ is unstable if $\bar{r} \in [0,\frac{\betap - \gamma}{\betap - \hatbetap})$~\cite[Lemma~2]{pagliara2018bistability}. 
Suppose $s(0) > \bar{s} := 1- \frac{\betap - \gamma}{\betap - \hatbetap} = \frac{\gamma - \hatbetap}{\betap - \hatbetap}$. Consequently, we have $r(0) <  \frac{\betap - \gamma}{\betap - \hatbetap}$. Since the possible IFE at the above initial condition are unstable, the vector field around $y=0$ points towards increasing value of $y$. Let $y(0) \neq 0$ and let $t_0, \epsilon_y > 0$ be suitable constants such that $y(t) \geq \epsilon_y$ for all $t \geq t_0$ until $s(t) > \bar{s}$. Thus, there exists $c > 0$ such that under \eqref{eq:siri_strong_switched}, 
$$ \dot{s} = -(\dot{y}+\dot{r}) \leq - c \betap \epsilon_y s \implies s(t) \leq s(0) e^{-c\betap\epsilon_y t} $$
until $s(t) > \bar{s}$. As a result, there exists $T_0$ such that $s(t) \leq \bar{s}$ for $t > T_0$. We exploit this property to prove convergence to a stable IFE. First we show that if $y(t) > \hat{y}^*_{\mr{int}}$, infected proportion eventually decreases. Indeed, we have
\begin{align*}
\dot{y} & = [\alpha \betap s + \alpha \hatbetap (1-s-y) - \gamma] y
\\ & = [\alpha (\betap - \hatbetap) s + \alpha \hatbetap (1-y) - \gamma] y
\\ & \leq [\alpha (\betap - \hatbetap) \bar{s} + \alpha \hatbetap (1-y) - \gamma] y
\\ & = [\alpha (\gamma - \hatbetap) + \alpha \hatbetap (1-y) - \gamma] y
\\ & = [- \alpha \hatbetap y - (1-\alpha) \gamma] y < 0
\end{align*}
since $\alpha \in (0,1)$. In other words, when $s(t) \leq \bar{s}$ for $t > T_0$, $\dot{y}(t) < 0$ for $y > \hat{y}^*_{\mr{int}}$, and as a result, the trajectory will eventually remain confined to \eqref{eq:siri_strong_lowy} and \eqref{eq:siri_strong_mediumy}.

Next we show that if $y(t) \in ({y}^*_{\mr{int}},\hat{y}^*_{\mr{int}})$, infected proportion eventually decreases as well. In this regime, we have
\begin{align*}
\dot{y} & = [\alpha \betap s + \hatbetap (1-s-y) - \gamma] y
\\ & = [(\alpha \betap - \hatbetap) s + \hatbetap (1-y) - \gamma] y
\\ & < [(\betap - \hatbetap) \bar{s} + \hatbetap(1-y) - \gamma] y
\\ & = [\gamma - \hatbetap + \hatbetap (1-y) - \gamma] = - \hatbetap y^2 < 0
\end{align*}
when $s(t) \leq \bar{s}$ for $t > T_0$. Finally, when the trajectory is eventually confined to \eqref{eq:siri_strong_lowy}, i.e., $y(t) \in (0,{y}^*_{\mr{int}})$, we have
\begin{align*}
\dot{y} & = [\betap s + \hatbetap (1-s-y) - \gamma] y
\\ & \leq [(\betap - \hatbetap) \bar{s} + \hatbetap (1-y) - \gamma] y
\\ & = [\gamma - \hatbetap + \hatbetap (1-y) - \gamma] y = - \hatbetap y^2 < 0.
\end{align*}
Therefore, when $s(t) \leq \bar{s}$, $y(t)$ is monotonically decreasing for all $y \in (0,1]$ and thus, the infected proportion asymptotically converges to an IFE.


\noindent {\bf Case 3: $\hatbetap [1-\hat{y}^*_{\mr{int}}]< \gamma < \hatbetap$}. For the dynamics \eqref{eq:siri_strong_lowy}, we have $R_0 = \frac{\betap}{\gamma} > 1$ and $R_1 = \frac{\hatbetap}{\gamma} > 1$ in this regime. Following Theorem \ref{theorem:siri_vanilla}, we conclude that all points in the IFE are unstable. Since $\mathbf{E3}$ is the endemic equilibrium for \eqref{eq:siri_strong_highy}, its existence as an endemic equilibrium for \eqref{eq:siri_strong_switched} requires 
$$ 1-\frac{\gamma}{\alpha \hatbetap} > \hat{y}^*_{\mr{int}} \iff \alpha \hatbetap[1-\hat{y}^*_{\mr{int}}] > \gamma ,$$
which is not satisfied in this parameter regime as $\alpha < 1$. 

Now observe that $\mathbf{E2}$ is the endemic equilibrium for both \eqref{eq:siri_strong_lowy} and \eqref{eq:siri_strong_mediumy}, but not for the dynamics \eqref{eq:siri_strong_highy}. Thus, for $\mathbf{E2}$ to be an equilibrium for the hybrid system \eqref{eq:siri_strong_switched}, we must have
$$ 0 < 1 - \frac{\gamma}{\hatbetap} < \hat{y}^*_{\mr{int}}, $$
or equivalently, $\hatbetap > \gamma$ and $\hatbetap [1-\hat{y}^*_{\mr{int}}]< \gamma$, i.e., precisely the parameter regime in this case. Further, for both \eqref{eq:siri_strong_lowy} and \eqref{eq:siri_strong_mediumy}, $R_1 = \frac{\hatbetap}{\gamma} > 1$, and as a result from Theorem \ref{theorem:siri_vanilla}, $\mathbf{E2}$ is locally stable. Let $y_{E2} = 1-\frac{\gamma}{\hatbetap}$.

It remains to show that the infected proportion $y(t)$ converges to the endemic level at $\mathbf{E2}$. Since all points on the IFE are unstable, the vector field at the IFE points towards increasing value of $y(t)$. Let $y(0) \neq 0$ and let $t_0, \epsilon_y > 0$ be suitable constants such that $y(t) \geq \epsilon_y$ for all $t \geq t_0$. Thus, there exists $c > 0$ such that under \eqref{eq:siri_strong_switched}, 
$$ \dot{s} = -(\dot{y}+\dot{r}) \leq - c \betap \epsilon_y s \implies s(t) \leq s(0) e^{-c\betap\epsilon_y t}. $$
As a result, for any $\epsilon_s > 0$, there exists $T_0$ such that $s(t) \leq \epsilon_s$ for $t > T_0$. In other words, $s(t)$ acts as a vanishing input to the dynamics of infected proportion. If $y(t) > \hat{y}^*_{\mr{int}}$,
\begin{align*}
\dot{y} & = [\alpha \betap s + \alpha \hatbetap (1-s-y) - \gamma] y
\\ & \leq [\alpha (\betap - \hatbetap) \epsilon_s + \alpha \hatbetap (1-y) - \hatbetap(1-y_{E2})] y
\\ & < [\alpha (\betap - \hatbetap) \epsilon_s + \hatbetap (y_{E2} - y)] y.
\end{align*}
Since the second term is strictly negative for $y > \hat{y}^*_{\mr{int}}$, there exists an $\epsilon_s$ and $T_0$ such that $\dot{y}(t) < 0$ for $t > T_0$, and as a result, the trajectory will eventually remain confined to \eqref{eq:siri_strong_lowy} and \eqref{eq:siri_strong_mediumy}. We now consider the following two sub-cases.



\noindent {\bf Case 3A: $y_{E2} \in (0,{y}^*_{\mr{int}})$.} Following analogous arguments as above, we first claim that there exists $T'_0$ such that $\dot{y}(t) < 0$ for $t > T'_0$, $y(t) \in ({y}^*_{\mr{int}},\hat{y}^*_{\mr{int}})$. In other words, the trajectory eventually remains confined to \eqref{eq:siri_strong_lowy}. Observe now that the dynamics of infected proportion in \eqref{eq:siri_strong_lowy} can be viewed as a perturbed system with $s(t) = 1 - r(t) - y(t)$ playing the role of a vanishing input. We now show that the equilibrium point $\mathbf{E2}$ is input-to-state stable (ISS) \cite{khalil}. Let $V(y) = (y - y_{E2})^2$ be the candidate ISS-Lyapunov function. Under the dynamics \eqref{eq:siri_strong_lowy}, we have
\begin{align*}
\dot{V}(y) & = 2(y-y_{E2}) [\betap s + \hatbetap r - \gamma] y 
\\ & = 2(y-y_{E2}) [\betap s + \hatbetap (1-s-y) - \hatbetap (1-y_{E2})] y
\\ & = 2(y-y_{E2}) [(\betap - \hatbetap) s - \hatbetap (y-y_{E2})] y
\\ & =  2(y-y_{E2}) (\betap - \hatbetap) ys - 2\hatbetap y V(y)
\\ & \leq - 2 (\hatbetap - \delta) y V(y),
\end{align*}
when $\rvert y(t) - y_{E2} \rvert \geq \delta^{-1} (\betap - \hatbetap) \rvert s(t) \rvert$. It follows from \cite{khalil} that $y_{E2}$ is ISS for the dynamics of infected proportion when $y(t) \in [0,{y}^*_{\mr{int}})$. As $s(t) \to 0$ as $t \to \infty$, we have $y(t) \to y_{E2}$.

\noindent {\bf Case 3B: $y_{E2} \in ({y}^*_{\mr{int}},\hat{y}^*_{\mr{int}})$.} {First observe that if $y(t) \in (0,{y}^*_{\mr{int}})$, we have
\begin{align*}
\dot{y} & = [\betap s + \hatbetap (1-s-y) - \gamma] y
\\ & = [(\betap - \hatbetap) s + \hatbetap (1-y) - \gamma] y
\\ & > [(\betap - \hatbetap) s + \hatbetap (1-{y}^*_{\mr{int}}) - \hatbetap (1-y_{E2})] y
\\ & = [(\betap - \hatbetap) s + \hatbetap (y_{E2}-{y}^*_{\mr{int}})] y > 0,
\end{align*}
since $y_{E2} > {y}^*_{\mr{int}}$ in this sub-case. Consequently, $y(t)$ remains confined to \eqref{eq:siri_strong_mediumy}. Let $V(y) = (y - y_{E2})^2$ as before. Under the dynamics \eqref{eq:siri_strong_mediumy}, we have
\begin{align*}
\dot{V}(y) & = 2(y-y_{E2}) \! [\alpha \betap s + \hatbetap (1-\! s-\! y) - \! \hatbetap (1-\! y_{E2})] y
\\ & = 2(y-y_{E2}) [(\alpha \betap - \hatbetap) s - \hatbetap (y-y_{E2})] y
\\ & =  2(y-y_{E2}) (\alpha \betap - \hatbetap) ys - 2\hatbetap y V(y)
\\ & \leq -2 (\hatbetap - \delta) y V(y)
\end{align*}
when $\rvert y(t) - y_{E2} \rvert \geq \delta^{-1} \rvert \alpha \betap - \hatbetap \rvert s(t)$ which implies that $y_{E2}$ is ISS for the dynamics of infected proportion when $y(t) \in ({y}^*_{\mr{int}},\hat{y}^*_{\mr{int}})$. Since $s(t) \to 0$ as $t \to \infty$, we have $y(t) \to y_{E2}$.}

\noindent {\bf Case 4: $\alpha \hatbetap [1-\hat{y}^*_{\mr{int}}] < \gamma < \hatbetap [1-\hat{y}^*_{\mr{int}}]$}. It follows from Case 3 above that for $\mathbf{E2}$ to be an equilibrium of the hybrid dynamics, we must have $\hatbetap [1-\hat{y}^*_{\mr{int}}]< \gamma$ which is not satisfied in this regime. Similarly, for $\mathbf{E3}$ to be an equilibrium, we must have $\alpha \hatbetap [1-\hat{y}^*_{\mr{int}}] > \gamma$ which is not satisfied in this regime. IFE being unstable follows from previous arguments as $R_0 > 1$ and $R_1 > 1$ for the dynamics \eqref{eq:siri_strong_lowy}. 


We now focus on the dynamics at the neighborhood of $(\hat{y}^*_{\mr{int}},1-\hat{y}^*_{\mr{int}})$. For $y = \hat{y}^*_{\mr{int}} + \epsilon_1, r = 1 - \hat{y}^*_{\mr{int}} - \epsilon_2$ with a sufficiently small $\epsilon_1,  \epsilon_2 > 0$, we have
\begin{align*}
\dot{r} & = [-\alpha \hatbetap (1-\hat{y}^*_{\mr{int}} - \epsilon_2) +\gamma] (\hat{y}^*_{\mr{int}} + \epsilon_1) > 0, 
\\ \dot{y} & = [\alpha \hatbetap (1-\hat{y}^*_{\mr{int}} - \epsilon_2) -\gamma] (\hat{y}^*_{\mr{int}} + \epsilon_1) < 0, 
\end{align*}
since $\gamma > \alpha \hatbetap [1-\hat{y}^*_{\mr{int}}]$. Similarly, at $y = \hat{y}^*_{\mr{int}} - \epsilon_1, r = 1 - \hat{y}^*_{\mr{int}} + \epsilon_2$ with a sufficiently small $\epsilon_1, \epsilon_2 > 0$, we have
\begin{align*}
\dot{r} & = [-\hatbetap (1-\hat{y}^*_{\mr{int}} + \epsilon_2) +\gamma] (\hat{y}^*_{\mr{int}} - \epsilon_1) < 0, 
\\ \dot{y} & = [\hatbetap (1-\hat{y}^*_{\mr{int}} + \epsilon_2) -\gamma] (\hat{y}^*_{\mr{int}} - \epsilon_1) > 0, 
\end{align*}
since $\gamma < \hatbetap [1-\hat{y}^*_{\mr{int}}]$. Since the differential inclusion in \eqref{eq:siri_strong_th2} is a convex combination of the dynamics in \eqref{eq:siri_strong_mediumy} and \eqref{eq:siri_strong_highy}, $(\hat{y}^*_{\mr{int}},1-\hat{y}^*_{\mr{int}})$ acts as a sliding mode of \eqref{eq:siri_strong_switched}. 

{Following identical arguments as {\bf Case 3B} above, it can be shown that when $y(t) \in (0,\hat{y}^*_{\mr{int}})$, $\dot{y} > 0$ and when $y(t) > \hat{y}^*_{\mr{int}}$, $\dot{y} < 0$ for some $t >  T_0$. Consequently, $y(t) \to {y}^*_{\mr{int}}$ as $t \to \infty$.}

\noindent {\bf Case 5: $\gamma < \alpha \hatbetap [1-\hat{y}^*_{\mr{int}}]$}. It follows from Case 3 that for $\mathbf{E3}$ to be an equilibrium of the hybrid dynamics, we require $\gamma < \alpha \hatbetap [1-\hat{y}^*_{\mr{int}}]$ and $\gamma < \alpha \hatbetap$, both of which are satisfied in this regime. Further, for \eqref{eq:siri_strong_highy}, $R_1 = \frac{\alpha \hatbetap}{\gamma} > 1$, and consequently, $\mathbf{E3}$ is locally stable. The IFE being unstable and $\mathbf{E2}$ not being an equilibrium point follows from previous arguments. 

Following identical arguments as {\bf Case 3B} above, it can be shown that when $y(t) < \hat{y}^*_{\mr{int}}$, $\dot{y} > 0$ and when $y(t) > \hat{y}^*_{\mr{int}}$, the endemic infection level at $\mathbf{E3}$ is ISS for the infection dynamics with $s(t)$ being the vanishing input. 

\end{proof}

\subsection{Proof of Proposition \ref{prop:equilibria-and-stability-weak-siri}}
\label{app:sec:last}

\begin{proof}
The proofs of statements $1-4$ follows from analogous arguments as the proof of Proposition \ref{prop:equilibria-and-stability-strong-siri} and is omitted. We now focus on establishing the claims regarding the stability of IFE. Note that local stability of the IFE depends on the dynamics \eqref{eq:siri_weak_lowy}. For statements $2-4$, we have $\gamma < \hatbetap$, and as a result, $R_1 = \frac{\hatbetap}{\gamma} > 1$ for \eqref{eq:siri_weak_lowy}. 

When $\gamma < \betap$, we have $R_0 = \frac{\betap}{\gamma} > 1$ for the dynamics \eqref{eq:siri_weak_lowy}. Following Case 3 of Theorem \ref{theorem:siri_vanilla}, all points in the IFE are unstable. 

When $\gamma > \betap$, we have $R_0 = \frac{\betap}{\gamma} < 1$ for the dynamics \eqref{eq:siri_weak_lowy}. It follows from Case 4 of Theorem \ref{theorem:siri_vanilla} that IFE with 
\begin{align*}
s^* > \frac{1-\frac{\hatbetap}{\gamma}}{\frac{\betap}{\gamma} - \frac{\hatbetap}{\gamma}} = \frac{\gamma - \hatbetap}{\betap - \hatbetap} \iff r^* < \frac{\betap - \gamma}{\betap - \hatbetap} 
\end{align*}
are locally stable and IFE with $r^* > \frac{\betap - \gamma}{\betap - \hatbetap}$ are unstable.
\end{proof}

\bibliographystyle{plainnat}  
\bibliography{arxiv_main}

\end{document}